\documentclass[orivec,envcountsect,UKenglish,cleveref, autoref, 11pt]{llncs}

\usepackage[singlespacing]{setspace}
\usepackage[papersize={8.5in,11in},left=1in,right=1in,top=1in,bottom=1in]{geometry}

\usepackage{hyperref}

\usepackage{comment}
\usepackage[utf8]{inputenc}

\usepackage{verbatim}
\usepackage{graphicx}
\usepackage{xcolor}
\usepackage{float}

\usepackage{amsthm}
\usepackage{amsfonts,amsmath,amssymb}
\usepackage[capitalise,noabbrev]{cleveref}
\usepackage{braket}
\usepackage{mathtools}

\usepackage{array}
\usepackage{multirow}
\usepackage[ruled,vlined,linesnumbered]{algorithm2e}
\usepackage{csquotes}
\usepackage{tabularx}

\usepackage{circuitikz}
\usepackage{tikz}
\usepackage{qcircuit}
\usepackage{mathtools}
\usepackage{pgfplots}

\usetikzlibrary{
    calc,
    trees,
    positioning,
    arrows,
    chains,
    shapes.geometric, 
    decorations.pathreplacing,
    decorations.pathmorphing,
    shapes, 
    matrix,
    shapes.symbols
    }

\newcommand{\bigO}{\mathcal{O}}

\newcommand{\F}{\mathbb{F}}
\newcommand{\R}{\mathbb{R}}
\newcommand{\E}{\mathbb{E}}

\newcommand{\Z}{\mathbb{Z}}
\newcommand{\N}{\mathbb{N}}

\renewcommand{\Pr}{\mathbb{P}}

\newlength{\strutdepth}%
\newcommand{\notes}[3]{
	\noindent{%
		\color{#1}{[#3]}\color{#1}}%
	\strut\vadjust{\kern-\strutdepth%
		\vtop to \strutdepth{%
			\baselineskip\strutdepth%
			\vss\llap{{\large\color{#1}\textbf{#2}\quad\color{black}}}\null%
		}%
	}%
}

\newcommand{\QFT}{\textnormal{QFT}}
\newcommand{\Simon}{\textsc{Simon}}
\newcommand{\Shor}{\textsc{Shor}}
\newcommand{\Period}{\textsc{Period}}
\newcommand{\EHS}{\textsc{Eker\aa -H\aa stad}}

\DeclareMathOperator{\id}{id}
\newcommand{\HD}{\{0,1\}}

\newtheorem*{lemma*}{Lemma}

\newcommand{\EM}{\textnormal{EM}}

\usepackage{cleveref}
\usepackage{subcaption}
\usepackage{url}

\definecolor{orcidlogocol}{HTML}{A6CE39}

\title{Quantum Period Finding is Compression Robust}

\titlerunning{Quantum Period Finding is Compression Robust}
\author{
	Alexander May\thanks{Funded by DFG under Germany's Excellence Strategy - EXC 2092 CASA - 390781972. 
	} \href{https://orcid.org/0000-0001-5965-5675}{\protect\includegraphics[height=\fontcharht\font`B]{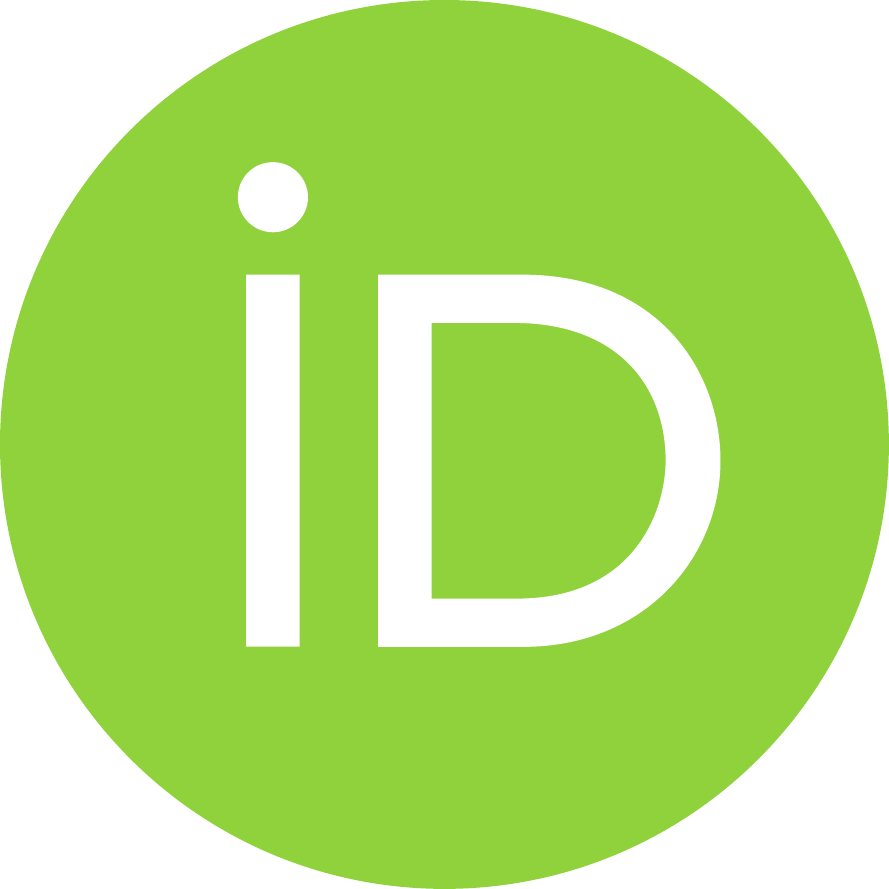}}\hspace*{1pt}
	\and
	Lars Schlieper$^\star$ \href{https://orcid.org/0000-0002-4870-1012}{\protect\includegraphics[height=\fontcharht\font`B]{orcid.pdf}}}
\institute{
	Horst G\"ortz Institute for IT Security \\ 
	Ruhr-University Bochum, Germany \\
	\email{\{alex.may,lars.schlieper\}@rub.de}}

\allowdisplaybreaks
\begin{document}
\maketitle

\begin{abstract}
We study quantum period finding algorithms such as Simon and Shor (and its variants Eker\aa-H\aa stad and Mosca-Ekert). For a periodic function $f$ these algorithms produce -- via some quantum embedding of $f$ -- a quantum superposition $\sum_x \ket{x}\ket{f(x)}$, which requires a certain amount of output qubits that represent $\ket{f(x)}$. We show that one can lower this amount to a single output qubit by hashing $f$ down to a single bit in an oracle setting.

Namely, we replace the embedding of $f$ in quantum period finding circuits by oracle access to several embeddings of hashed versions of $f$. We show that on expectation this modification only doubles the required amount of  quantum measurements, while significantly reducing the total number of qubits. For example, for Simon's algorithm that finds periods in $f: \F_2^n \rightarrow \F_2^n$ our hashing technique reduces the required output qubits from $n$ down to $1$, and therefore the total amount of qubits from $2n$ to $n+1$. We also show that Simon's algorithm admits real world applications with only $n+1$ qubits by giving a concrete realization of a hashed version of the cryptographic Even-Mansour construction. Moreover, for a variant of Simon's algorithm on Even-Mansour that requires only classical queries to Even-Mansour we save a factor of (roughly) 4 in the  qubits.

Our oracle-based hashed version of the
Eker\aa-H\aa stad algorithm for factoring $n$-bit RSA reduces the required qubits from $(\frac 3 2 + o(1))n$ down to $(\frac 1 2 + o(1))n$. 
We also show a real-world (non-oracle) application in the discrete  logarithm setting by giving a concrete realization of a hashed version of Mosca-Ekert for the Decisional Diffie Hellman problem in $\F_{p^m}$, thereby reducing the number of qubits by even a linear factor from $m \log p$ downto $\log p$.

\end{abstract}

\section{Introduction}

Throughout this paper, we consider only logical qubits that are error-free.
Although there is steady progress in constructing larger quantum computers, within the next years the number of qubits seems to be too limited for tackling problems of interesting size, e.g. for period finding applications in cryptography \cite{DBLP:conf/isita/KuwakadoM12,DBLP:conf/crypto/KaplanLLN16,DBLP:conf/asiacrypt/Leander017,DBLP:journals/qic/SantoliS17,DBLP:conf/asiacrypt/RoettelerNSL17,DBLP:journals/qic/HanerRS17,DBLP:journals/ieeesp/RoettelerS18}. 

Shor's algorithm~\cite{DBLP:conf/focs/Shor94} for polynomial time factorization of $n$-bit numbers computes a superposition $\sum_x \ket{x}\ket{f(x)}$ with $2n$ {\em input qubits}  representing the input $\ket{x}$ to $f$ and {\em $n$ output qubits} representing the output $\ket{f(x)}$ of the function.

However, it may not be necessary to implement a full-fledged $3n$-qubit Shor algorithm in order to factor numbers or compute discrete logarithms. Quantum computers with a very limited number of qubits might still serve as a powerful oracle that assists us in speeding up classical computations. For instance, Bernstein, Biasse and Mosca~\cite{DBLP:conf/pqcrypto/BernsteinBM17} developed an algorithm that factors $n$-bit numbers with the help of only a sublinear amount of $n^{\frac 2 3}$ qubits in subexponential time that is (slightly) faster than the currently best known purely classical factorization algorithm.

Several other algorithms saved on the number of qubits in Shor's algorithm by shifting some more work into a classical post-processing, while -- in contrast to~\cite{DBLP:conf/pqcrypto/BernsteinBM17} -- still preserving polynomial run time. Interestingly, all these algorithms concentrate on  reducing the input qubits, while keeping $n$ output qubits. Seifert~\cite{DBLP:conf/ctrsa/Seifert01} showed that -- using for the classical post-process simultaneous Diophantine approximations instead of continued fractions -- the number of input qubits can be reduced from $2n$ to $(1+o(1))n$. For $n$-bit RSA numbers, which are a product of two $n/2$-bit primes, Eker\aa\ and H\aa stad \cite{DBLP:journals/corr/EkeraH17} reduced the number of input qubits down to $(\frac 1 2 + o(1))n$,  using  some variant of the Hidden Number Problem~\cite{DBLP:conf/soda/BonehV97} in the post-process. Thus, the Eker\aa -H\aa stad version of Shor's algorithm factors $n$-bit RSA with a total of $(\frac 3 2 + o(1))n$ qubits. 

Mosca and Ekert~\cite{DBLP:conf/qcqc/MoscaE98} showed that one can reduce the number of input qubits even down to a single one, at the cost of an increased depth of Shor's quantum circuit. 

\hskip 0.1cm

\noindent {\bf Our contribution.} We hash $f(x)$ in the output qubits down to $t$ qubits, where $t$ can be as small as $1$. This can be realized using quantum embeddings of $h \circ f$ for different hash functions $h$, for which we assume oracle access. Our basic observation is that hashing preserves the periodicity of $f$. Namely, if $f(x) = f(x+s)$ for some period $s$ and all inputs $x$ then also 
\[
  h(f(x)) = h(f(x+s)) \textrm{ for the period $s$ and all inputs $x$}.
\] 
The drawback of hashing is that $h$ certainly introduces many more undesirable collisions $h(f(x)) = h(f(x'))$ where $x,x'$ are not a multiple of $s$ apart. Surprisingly, even for 1-bit range hash functions this plethora of undesirable collisions  does not at all affect the correctness of our hashed quantum period finding algorithms, and only insignificantly increases their runtimes.

More precisely, concerning correctness we show that a replacement of $f$ by some hashed version of $f$ has the following effects.

\begin{description}
\item[Simon's algorithm:] In the input qubits, we still measure only vectors $y$ that are orthogonal to the period $s$. The amplitudes of all other inputs cancel out.
\item[Shor's algorithm:] Let the period be $d=2^r$, and let us use $q>r$ input qubits. Then we still measure in the input qubits only numbers $y$  that are multiples $2^{q-r}$. The amplitudes of all other inputs cancel out. In the case of general (not only power of two) periods we measure all inputs $y\not=0$ with exactly half the probability as without hashing. 
\end{description} 

Our correctness property immediately implies that the original post-processing in Simon's algorithm (Gaussian elimination) and in Shor's algorithm (e.g. continued fractions) can still be used in the hashed version of the algorithms for period recovery. 

However, this does not automatically imply that we achieve similar runtimes. Namely, in the original algorithms of Simon and Shor we measure all $y$ having a non-zero amplitude with a uniform probability distribution. In Simon's algorithm for some period $s \in \F_2^n$ we obtain each of the $2^{n-1}$ many $y \in \F_2^n$ orthogonal to $s$ with probability $\frac{1}{2^{n-1}}$. In Shor's algorithm with period $d=2^r$, we measure each of the $d$ many possible multiples $y$ of $2^{q-r}$ with probability $\frac{1}{d}$. 

These uniform probability distributions are destroyed by moving to the hashed version of the algorithms. Since $h(f(x)) = h(f(x'))$ for $x \not= x'$ happens for universal 1-bit range hash functions with probability $\frac 1 2$, the undesirable collisions put a probability weight of (roughly) $\frac 1 2$ on measuring $y = 0$ in the input qubits.

This seems to be bad news, since neither in Simon's algorithm does the zero vector $y$ provide information about $s$, nor does in Shor's algorithm the zero-multiple $y$ of $2^{q-r}$ provide information about $d$. 
However as good news, we show that besides putting probability weight $\frac 1 2$ on $y=0$, hashing does not destroy the probability distribution stemming from the amplitudes of quantum period finding algorithms. Namely, we show that for the whole class of quantum period finding circuits that we consider -- including Simon, Shor (and its variants Eker\aa -H\aa stad, Mosca-Ekert) -- the following result holds:  If the probability to measure $y$ is $p(y)$ when using $f$, then we obtain probability $p(y)/2$ to measure $y$ when using $h \circ f$,  where the latter probability is taken over the random choice of $h$ from a family of 1-bit range universal hash functions.

Put differently, if we condition on the event that we do {\em not} measure $y=0$ in the input bits (which happens in roughly every second measurement) in both cases  -- using $f$ itself or its hashed version $h \circ f$  -- we obtain exactly the same probability distribution for the measurements of any $y \not=0$. This implies that our hashing approach preserves not only the correctness but also the runtime analysis of any processing of the measured data in a classical post-process. Thus, at the cost of only twice as many quantum measurements we save all but one of the output qubits. More generally, we show that at the cost of $\frac 1 {1-2^{-t}}$-times more measurements we may compress to $t$ output qubits.

\medskip

In particular, we show that the original Simon algorithm \cite{DBLP:conf/focs/Simon94} --- that recovers for a periodic function $f: \F_2^n \rightarrow \F_2^n$ its period in time polynomial in $n$ with expected $n+1$ measurements using $2n$ qubits --- admits an {\em oracle-based} hashed version with expected $2(n+1)$ measurements using only $n+1$ qubits. Moreover, we show that this leads to an {\em explicit} (non-oracle, efficiently constructable) realization of the quantum Even-Mansour attack~\cite{DBLP:conf/isita/KuwakadoM12,DBLP:conf/crypto/KaplanLLN16} with only $n+1$ qubits. For the quantum attack~\cite{DBLP:conf/asiacrypt/BonnetainHNSS19} on Even-Mansour with only classical access to the cipher, called Offline-Simon, we provide an {\em explicit} hashed realization that saves even (roughly) a factor of $4$ in the number of qubits.

The original Eker\aa -H\aa stad version of Shor's algorithm that computes discrete logarithms $d$ in some abelian group $G$ in polynomial time using $(1+o(1)) \log d + \log(|G|)$ qubits requires in its {\em oracle-based} hashed version only $(1+o(1)) \log d$ qubits. Moreover, the Eker\aa -H\aa stad algorithm computes the factorization of an RSA modulus $N=pq$ of bit-size $n$ in time polynomial in $n$ using $(\frac 3 2 + o(1))n$ qubits, whereas our {\em oracle-based} hashed version reduces this to only $(\frac 1 2 + o(1))n$ qubits. We leave it as an open problem whether there exist an {\em explicit} hashed Eker\aa -H\aa stad realization. For Eker\aa -H\aa stad, one has to compute hashed versions of the exponentiation function $f_{a,N}: x \mapsto a^x \bmod N$. Notice that it is of course {\em not sufficient} to compute $f_{a,N}$ first, and afterwards hash the result, since this would require qubits for representing the full range of $f_{a,N}$. 

As a positive result in this direction, we show that exponentiation functions in certain cases indeed admit {\em explicit} realizations. 
More precisely, we provide an {\em explicit} realization of a hashed version of Mosca-Ekert in the discrete logarithm setting, where we reduce the number of qubits to solve the Decisional Diffie Hellman problem in the subgroup of quadratic residues of $\F_{p^m}$ by a linear factor $\Theta(m)$ from $1 + m \cdot \log p$ qubits downto only $1+ \log p$ qubits. Our {\em explicit} hashed Mosca-Ekert realization provides a good example that universal hash function families are in general not necessary for our hash technique, since we show correctness for our Diffie Hellman application using a single fixed hash function.

We believe that the universal hash function family property might be relaxed in other realizations as well. We conjecture that often in practice a single $h$ should still work. Even choosing $h$ simply as the projection of $f$ to a single bit should work for most functions of interest. 
We believe that it is of theoretical and practical interest to study in more generality,  which classes of $f$ admit a memory-efficient computation of their hashed versions.

\hskip 0.1cm

Our paper is organized as follows. 
In \cref{hashed_simon} we present our first main result that Simon's algorithm is compression robust.
In \cref{sec:even_mansou} we provide an explicit hashed Simon realization of Even-Mansour, and transfer this in \cref{sec:offline_simon} to an explicit hashed Offline-Simon realization.

%
Subsequently, we generalize the hash concept to Shor's algorithm and a more general class of period-finding circuits. 
For didactic reasons, we study in \cref{hashed_shor} first the simple case of Shor's algorithm for periods that are a power of two. In \cref{shor_general}, as our second main result we generalize  to any quantum circuits that fall in our period finding class. As a consequence, in \cref{shor_general} we obtain a hashed version of Shor's algorithm with general periods, and in \cref{ekera_sec} a hashed version of Eker\aa -H\aa stad. In \cref{sec:mosca_ekert}, we describe additional properties of our hash function family that admit a hashed version of Mosca-Ekert, and provide an {\em explicit} Mosca-Ekert realization for Decisional Diffie Hellman in $\F_{p^m}$ in \cref{sec:ddh}. 



%
%
%
%
\section{Preliminaries on Period Finding Algorithms}
\label{prelim_sec}
%
%
Let us first recall some quantum notation. The reversible {\em quantum embedding} of a classical function $f$ is defined as 
\[
  U_f: \ket{x}\ket{y}\mapsto\ket{x}\ket{y+f(x)}.
\]
The $1$-qubit {\em Hadamard gate} realizes the mapping $H_1: \ket{x} \mapsto\frac{1}{\sqrt{2}}\sum_{y\in\F_2}(-1)^{xy}\ket{y}$. Its $n$-qubit version is defined as the $n$-fold tensor product $H_n = \bigotimes_{i=1}^n H_1$. The $n$-qubit {\em Quantum Fourier Transform} (QFT) is the mapping
\[
  \QFT_n:	\ket{x}\mapsto\frac{1}{\sqrt{2^n}}\sum_{y\in\F_2^n}e^{2\pi i\frac{x}{2^n}y}\ket{y}.
\] 
Notice that $\QFT_1=H_1$.

\begin{definition}
\label{universal}
	A hash function family $\mathcal{H}_t:=\{h: \mathcal{D} \to \HD^t\}$ is {\em universal} if for all $x,y\in\mathcal{D}$, $x\not=y$ we have
\[
  \mathbb{P}_{h\in\mathcal{H}_t}\left[h(x)=h(y)\right]=\frac{1}{2^t}\;.
\]
\end{definition}
\noindent Efficient instantiations of (homomorphic) universal hash function families exist, e.g.
  \begin{equation}\label{eq:hash}
	\mathcal{H}_t=\{h_r:\F_2^n\to\F_2^t \;|\; r\in \left(\F_2^n\right)^t,\ h_r(x)=(\langle x , r_1 \rangle,\ldots,\langle x , r_t \rangle) \}\;. 
  \end{equation}
It is easy to see that {\em strongly 2-universal} hash function families as defined in~\cite{DBLP:books/daglib/0012859} are universal in the sense of \cref{universal}.
%
%

%
%

\section{Hashed-Simon}\label{hashed_simon}

Let us briefly recall Simon's original algorithm. Let $f: \F_2^n \rightarrow \F_2^n$ be periodic with period $s \in \F_2^n \setminus \{0^n\}$, that is $f(x) = f(x+s)$ for all $x \in \F_2^n$. We call $f$ a {\em Simon function} if it defines a $(2:1)$-mapping, i.e.
\[
  f(x) = f(y) \Leftrightarrow (y = x) \textrm{ or } (y=x+s).
\]
The use of Simon functions allows for a clean theoretical 
analysis, although Simon's algorithm works also for more general periodic functions as shown in~\cite{DBLP:conf/soda/AlagicMR07,articleQuantumAlgorithmsforAlgebraicProblems,DBLP:conf/asiacrypt/Leander017}. For ease of notation, we restrict ourselves to Simon functions.

The {\em Simon circuit} $Q_f^{\Simon}$ from \cref{circuit:simon} uses $n$ input and $n$ output qubits for realizing the embedding of $f$. It can easily be shown that in the $n$ input qubits we measure only $y \in \F_2^n$ such that $y \perp s$, i.e. $ys = 0$.  

\begin{figure}[h]
	\centering
	\mbox{
		\Qcircuit @C=1em @R=1em {
			\lstick{\ket{0^n}}	& \qw & \gate{H_n} & \qw & \multigate{1}{U_{f}} & \qw & \gate{H_n} 	& \qw & \meter\\
			\lstick{\ket{0^n}}	& \qw & \qw	 	   & \qw & \ghost{U_{f}}        & \qw & \qw	     	& \qw & \qw
		}
	}
	\caption{Quantum circuit $Q^{\Simon}_{f}$}\label{circuit:simon}
\end{figure}
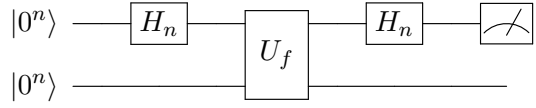

The \textsc{Simon} algorithm uses $Q_f^{\Simon}$ until we have collected $n-1$ linearly independent vectors, from which we compute the unique vector $s$ that is orthogonal to all of them using Gaussian elimination in time $\bigO(n^3)$.

Our \textsc{Hashed-Simon} (Algorithm~\ref{alg:simon}) is identical to the \textsc{Simon} algorithm with the only difference that $Q_f^{\textsc{Simon}}$ is replaced by $Q_{h \circ f}^{\textsc{Simon}}$, where in each iteration we instantiate $Q_{h \circ f}^{\textsc{Simon}}$ with some hash function $h$ freshly drawn from a universal $t$-bit range hash function family $ \mathcal{H}_t$. Notice that \textsc{Simon} can be considered as special case of \textsc{Hashed-Simon}, where we choose $t=n$ and the identity function $h=\textrm{id}$. This slightly abuses notation, since $\mathcal{H}_n=\{\textrm{id}\}$ is not universal. However, the following \cref{orthogonality} holds without universality of $\mathcal{H}_n$. In \cref{orthogonality} we show the correctness property of \textsc{Hashed-Simon} that by replacing $Q_{f}^{\textsc{Simon}}$ with $Q_{h \circ f}^{\textsc{Simon}}$, we still measure only $y$ orthogonal to $s$. 
\vspace*{-0.5cm}
\begin{algorithm}[ht]
	\DontPrintSemicolon
	\SetAlgoLined
	\SetKwInOut{Input}{Input}\SetKwInOut{Output}{Output}
	
	\SetKwComment{COMMENT}{$\triangleright$\ }{}
	
	\Input{ Simon function $f: \F_2^n \rightarrow \F_2^n$, universal $\mathcal{H}_t:= \{h: \F_2^n \rightarrow \F_2^t\}$}
	\Output{ Period $s$ of $f$}
	
	\Begin{
		Set $Y=\emptyset$.\; 
		\Repeat{$Y$ contains $n-1$ linear independent vectors}{
			Run $Q^{\Simon}_{h \circ f}$ on $\ket{0^n}\ket{0^t}$ for some freshly chosen $h \in_R \mathcal{H}_t$. \;
			Let $y$ be the measurement of the $n$ input qubits. \;
			If $y \notin \textrm{span}(Y)$, then include $y$ in $Y$.
		}
		Compute $\{s\} = Y^\perp \setminus \{0^n\}$ via Gaussian elimination.\;
		\KwRet{$s$.}
	}
	\caption{\textsc{Hashed-Simon}\label{alg:simon}}
\end{algorithm}

\begin{lemma}[Orthogonality]\label{orthogonality} 
	Let $f:\F_2^n \rightarrow \F_2^n$ be a Simon function with period $s$. Let $h:\F_2^n \rightarrow \F_2^t$ and $f_h = h \circ f: \F_2^n \rightarrow \F_2^t$. Let us apply $Q^{\Simon}_{h \circ f}$ on $\ket{0^n}{\ket{0^t}}$. 
	Then we obtain superposition 	
	\begin{equation*}	
		\sum_{f_h(x)\in\textrm{Im}(f_h)} \sum_{\underset{y\perp s}{y\in \mathbb{F}_2^n}} w_{y,f_h(x)} \Ket{y} \ket{f_h(x)}
		\textrm{ where }
		w_{y,f_h(x)} = \frac{1}{{2^{n}}}\sum_{x\in f_h^{-1}(f_h(x))} (-1)^{xy}\;.
	\end{equation*}	 
\end{lemma}

\begin{proof}
	Since $f$ is a Simon function we have $f(x) = f(x+s)$ and therefore $f_h(x) = f_h(x+s)$. This implies $x \in f_{h}^{-1}(z)$ iff $x+s \in f_{h}^{-1}(z)$. 
	
	An application of $Q_{h\circ f}$ on input $\ket{0^n}{\ket{0^t}}$ yields for the operations $H_n\otimes I_t$ and $U_{f_{h}}$ 
	\[
	\Ket{0^n}\Ket{0^t} 
	\mathmakebox[12mm][c]{\overset{H_n\otimes I_t}{\rightarrow}}
	\frac{1}{2^{n/2}}\sum_{x \in \mathbb{F}_2^n} \Ket{x}\Ket{0^t}
	{\overset{U_{f_{h}}}{\rightarrow}}
	\frac{1}{2^{n/2}}\sum_{x \in \mathbb{F}_2^n} \Ket{x}\Ket{f_{h}(x)}.
	\]
	Using $x \in f_{h}^{-1}(z)$ iff $x+s \in f_{h}^{-1}(z)$, we obtain
	\[
	\frac{1}{2^{n/2}}\sum_{x \in \mathbb{F}_2^n} \Ket{x}\Ket{f_{h}(x)}
	{=}
	\frac{1}{2^{n/2}}\sum_{x \in \mathbb{F}_2^n} \frac{1}{2}\left(\Ket{x}+\Ket{x+s}\right)\Ket{f_{h}(x)}
	\]
	An application of $H_n$ now yields
	\begin{align*}
	& \frac{1}{2^{n}}\sum_{x \in \mathbb{F}_2^n}\sum_{y \in \mathbb{F}_2^n} \frac{1}{2}\left((-1)^{xy}+(-1)^{(x+s)y}\right)\Ket{y}\Ket{f_{h}(x)}\\
	{=}&
	\frac{1}{2^{n}}\sum_{x \in \mathbb{F}_2^n}\sum_{y \in \mathbb{F}_2^n} \frac{1}{2}(-1)^{xy}\left(1+(-1)^{sy}\right)\Ket{y}\Ket{f_{h}(x)}\\
	{=}&
	\frac{1}{2^{n}}\sum_{x \in \mathbb{F}_2^n}\sum_{\underset{y\perp s}{y\in \mathbb{F}_2^n}} (-1)^{xy}\Ket{y}\Ket{f_{h}(x)}
	\;\textnormal{.}
	\end{align*}
	The statement of the lemma follows.
\end{proof}

	\noindent From \cref{orthogonality}'s superposition
	\begin{equation}
	\label{amplitude}
		\sum_{f_h(x)\in\textrm{Im}(f_h)} \sum_{\underset{y\perp s}{y\in \mathbb{F}_2^n}}\underbrace{\frac{1}{{2^{n}}}\sum_{x\in f_h^{-1}(f_h(x))} (-1)^{xy}}_{w_{y,{f_h(x)}}} \Ket{y} \ket{f_h(x)}
	\end{equation}	
	we see that only $y \in \F_2^n$ with $y \perp s$ have a non-vanishing amplitude $w_{y,f_h(x)}$.
	
		Assume that we measure some fixed $z = f_h(x) \in \{0,1\}^t$ in the output qubits. Then an easy calculation shows that \cref{amplitude} collapses  to
		\begin{equation}
		\label{amplitude_fix}
		\sum_{\underset{y\perp s}{y\in \mathbb{F}_2^n}} \underbrace{\frac{1}{(2^{n}\cdot|f_h^{-1}({z})|)^{1/2}}\sum_{x \in f_h^{-1}({z})}(-1)^{xy}}_{w_{y,{z}}} \Ket{y} \Ket{{z}}\;.
		\end{equation}			

	Recall that \Cref{orthogonality} contains the analysis of Simon's original algorithm as the special case $h=\textrm{id}$. In this case, we know that by the definition of a Simon function $|f_h^{-1}(z) |=2$ for all $z$ and $\sum_{x \in f_h^{-1}({z})}(-1)^{xy} \in \{\pm 2\}$.
Thus, all $y \perp s$ have amplitude $\pm \frac{1}{2^{(n-1)/2}}$. This means that a measurement yields the uniform distribution over all $y \perp s$.

The following lemma will be useful, when we analyze superpositions over all $z$.

\begin{lemma}\label{Simon_Sum_Zero}
	Let $f:\F_2^n\mapsto\F_2^\ell$ be a Simon function. Then $\sum_{z\in\F_2^\ell} w_{y,z} = 0$
	for all $y\neq 0 $.
\end{lemma}

\begin{proof}	
	Fix $y\not=0^n$. If $y\not\perp s$ then all $w_{y,z}=0$ and thus the claim follows. 
	Hence, in the following let $y \perp s$. If $z \notin f(\F_2^n)$ then $w_{y,z}=0$. Therefore 
	\[
	\sum_{z \in \F_2^\ell} w_{y,z} = \sum_{z\in f(\F_2^n)} w_{y,z}.
	\]
	Since $f$ is a Simon function, $f$ is a (2:1)-mapping. Thus
	\[
	\sum_{z\in f(\F_2^n)} w_{y,z} = \frac{1}{2}\sum_{x\in \F_2^n} w_{y,f(x)}.
	\]
	Using the definition of $w_{y,f(x)}$ in Eq.~(\ref{amplitude}) with $h=\textrm{id}$ yields
	\[
	\frac{1}{2}\sum_{x\in \F_2^n} w_{y,f(x)} = \frac{1}{{2^{n+1}}}\sum_{x\in \F_2^n}  (-1)^{xy}\;.
	\]
	Since for $y\not=0$ we have $\sum_{x\in \F_2^n}  (-1)^{xy} = 0$, the claim follows.
\end{proof}

Let us now develop some intuition for the amplitudes $w_{y,z}$ in Eq.~(\ref{amplitude_fix}) for $1$-bit range hash functions $h: \F_2^n \rightarrow \F_2$. We expect that $|f_h^{-1}(z)| \approx 2^{n-1}$. We first look at the amplitude $w_{0^n, z}$ of $\ket{y} = \ket{0^n}$. Since for all $x \in \F_2^n$ we have $(-1)^{xy}=1$, the amplitude  of $\ket{0^n}$ adds up to $w_{0^n, z} = \left( \frac{f_h^{-1}(z)}{2^n} \right)^{\frac 1 2} \approx \frac 1 {\sqrt{2}}$. Hence, we expect to measure the zero-vector $0^n$ with probability approximately $\frac 1 2$. This seems to be bad news, since the zero-vector is the only one orthogonal to $s$ that does not provide any information about $s$. 

However, we show that all $y \perp s$ with $y \not= 0^n$ still appear with significant amplitude. Intuitively, $\sum_{x \in f_h^{-1}(z)}(-1)^{xy}$ describes for $y \not=0^n$ a random walk with $| f_h^{-1}(z) |$ steps. Thus, this term should contribute on expectation roughly  $|f_h^{-1}(z)|^{\frac 12}$ to the amplitude of $\ket{y}$. So we expect for all $y \perp s$ with $y \not= 0^n$ an amplitude of 
\[ 
  w_{y,z} = \frac{\sum_{x \in f_h^{-1}({z})}(-1)^{xy}}{(2^{n}\cdot|f_h^{-1}({z})|)^{1/2}} \approx \frac{1}{2^{n/2}}. 
\] 
This in turn implies that conditioned on the event that we do not measure $0^n$ (which happens with probability roughly $\frac 1 2$), we still obtain the uniform distribution over all remaining $y \perp s$.  

We make our intuition formal in the following theorem.

\begin{theorem}\label{measuring_probabilities}
	Let $\mathcal{H}_t=\{h: \F_2^n \rightarrow \F_2^t\} $ be universal, and let $f$ be a Simon function with period $s$. Then we measure in Algorithm \textsc{Hashed-Simon} in the first $n$ qubits any $y \perp s, y \not=0$ with probability $\frac{1-2^{-t}}{2^{n-1}}$, where the probability is taken over the random choice of $h \in \mathcal{H}_t$.
\end{theorem}

	\begin{proof} 
		From \Cref{orthogonality} in the case $h = \textrm{id}$, we conclude that \textsc{Simon} gives us a superposition
		\begin{equation*}
			\sum_{f(x)\in\textrm{Im}(f)} \sum_{\underset{y\perp s}{y\in \mathbb{F}_2^n}} w_{y,f(x)} \Ket{y} \ket{f(x)}
			\textrm{ where }
			w_{y,f(x)} = \frac{1}{{2^{n}}}\sum_{x\in f^{-1}(f(x))} (-1)^{xy}\;.
		\end{equation*}
		For ease of notation let us denote $z = f(x)$. In particular for $z\notin \textrm{Im}(f)$ we have $w_{y,z}=0$.
		We measure any $y$ with probability $\sum_{z \in \F_2^n} |w_{y,z}|^2$.
		
		Let $p(y)$ denote the probability to measure $y$ in the first $n$ qubits. Since $w_{y,z} \in \R$, we obtain $|w_{y,z}|^2 = w_{y,z}^2$ and hence the identity
		\begin{equation}
			\label{simon-amplitude}
			p(y) = \sum_{z \in \F_2^n} w_{y,z}^2 = \frac{1}{2^{n-1}}\;.
		\end{equation}
		Let us now look at \textsc{Hashed-Simon} with a $t$-bit range hash function $h \in \mathcal{H}_t$. From \cref{orthogonality} we get
		\[
		\ket{\Phi_h}=	\sum_{z' \in \F_2^t} \sum_{\underset{y\perp s}{y\in \mathbb{F}_2^n}} w_{y,z'} \ket{y} \ket{z'}\;.
		\]
		With respect to the amplitudes $w_{y,z}$ of \textsc{Simon} the superposition of \textsc{Hashed-Simon} can be written as
		\[
		\ket{\Phi_h}=	\sum_{z' \in \F_2^t} \sum_{\underset{y\perp s}{y\in \mathbb{F}_2^n}} \left(\sum_{z\in h^{-1}(z')} w_{y,z}\right) \ket{y} \ket{z'}\;,
		\]	
		where
		\[
		\sum_{z\in h^{-1}(z')} w_{y,z}
		= 
		\sum_{z\in h^{-1}(z')}\frac{1}{{2^{n}}}\sum_{x\in f^{-1}(z)}(-1)^{xy}
		= 
		\frac{1}{{2^{n}}}\sum_{x\in f_h^{-1}(z')}(-1)^{xy}
		=
		w_{y,z'}\;.
		\]
		Let us denote by $p_h(y) = \Pr_{h\in\mathcal{H}_t}[y]$ the probability that we measure $y$ in the first $n$ qubits when applying $Q^{\Simon}_{h \circ f}$. Our goal is to show that $p_h(y) = (1-2^{-t})\cdot p(y) = \frac{1-2^{-t}}{2^{n-1}}$.
		
		For some $h \in \mathcal{H}_t$ we denote $I_{h,z'}=\{ z \in f(\F_2^n) \mid h(z)=z'\}$. Since $\mathbin{\dot{\bigcup}}_{z'\in \F_2^t} I_{h, z'} = f(\F_2^n)$ and $w_{y,z} \in \R$, \textsc{Hashed-Simon} yields	
		
		\begin{align}
			p_h(y)
			& = \frac{1}{|\mathcal{H}_t|}\sum_{h\in\mathcal{H}_t}\sum_{z'\in \F_2^t}\left|\sum_{z\in I_{h,z'}} w_{y,z}\right|^2\nonumber \\
			& = \frac{1}{|\mathcal{H}_t|}\sum_{h\in\mathcal{H}_t}\sum_{z'\in \F_2^t}\left(\sum_{z\in I_{h,z'}} w_{y,z}\right)^2\;. 
			\label{cross_g}
		\end{align}
		In Eq.~(\ref{cross_g}) we obtain a cross-product $w_{y,z_1}w_{y,z_2}$ for $z_1 \not=z_2$ iff $z_1, z_2$ are in the same set $I_{h,z'}$, i.e. iff $h(z_1)=h(z_2)$. Using \cref{universal} of a universal hash function family, we obtain $\Pr_{h \in \mathcal{H}_t}[h(z_1)=h(z_2)]=2^{-t}$ for any $z_1 \not= z_2$. This implies that for exactly $2^{-t}$ of all $h \in \mathcal{H}_t$ we obtain $h(z_1)=h(z_2)$. 
		
		Further using $w_{y,z}=0$ for $z \notin f(\F_2^n)$, we conclude that
		\[
		p_h(y) =\sum_{z \in \F_2^n} w_{y,z}^2+2^{-t}\sum_{z_1\not=z_2} w_{y,z_1}w_{y,z_2}\;.
		\]
		From \cref{Simon_Sum_Zero} we know that 
		\[
		0 = 2^{-t} \left( \sum_{z \in \F_2^n} w_{y,z} \right)^2 = 2^{-t} \sum_{z \in \F_2^n} w_{y,z}^2+  2^{-t} \sum_{z_1\not=z_2} w_{y,z_1}w_{y,z_2}\;.
		\] 
		An application of this identity together with Eq.~(\ref{simon-amplitude}) gives us
		\[
		p_h(y) = (1-2^{-t}) \sum_{z \in \F_2^n}  w_{y,z}^2 = (1-2^{-t})\cdot p(y) = \frac{1-2^{-t}}{2^{n-1}} 
		\;.\vspace*{-0.75cm}
		\]	
	\end{proof}

	\begin{corollary}
		We measure in Algorithm \textsc{Hashed-Simon} in the first $n$ qubits $y =0$ with probability 
		$2^{-t}+(1-2^{-t})\cdot2^{1-n}$.
	\end{corollary}

\begin{theorem}\label{main_theorem_simon}
	Let $\mathcal{H}_t=\{h: \F_2^n \rightarrow \F_2^t\} $ be universal, and let $f:\F_2^n \rightarrow \F_2^n$ be a Simon function with period $s \in \F_2^n$.  \textsc{Hashed-Simon} recovers $s$ with expected $\frac{n+1}{1-2^{-t}}$ applications of quantum circuits $Q^{\Simon}_{h \circ f}$, $h \in_R \mathcal{H}_t$, that use only $n+t$ qubits. 
\end{theorem}

	\begin{proof}
		Let us define a random variable $X_i$, $1 \leq i < n$ for the number of applications of $Q^{\Simon}_{h \circ f}$ until \textsc{Hashed-Simon} finds $i$ linearly independent $y_1 \ldots y_i$. Let $E_i$ be the event that we already have $i-1$ linearly independent $Y=\{y_1, \ldots, y_{i-1}\}$ and we measure some $y_i \notin \textrm{span}(Y)$. Define $p_i = \Pr[E_i]$. Using \cref{measuring_probabilities}, we obtain 
		\[
		p_1 =(1-2^{-t})\cdot  \frac{2^{n-1}-1}{2^{n-1}}. 
		\]
		Since $\left\vert\textrm{span}\{y_1 \ldots, y_{i-1}\}\right\vert=2^{i-1}$, we obtain from  \cref{measuring_probabilities} more generally 
		\[
		p_i =(1-2^{-t})\cdot  \frac{2^{n-1} - 2^{i-1}}{2^{n-1}}.
		\]
		Clearly, $X_i$ is geometrically distributed with parameter $p_i$. Let\linebreak $X=X_1+ \ldots + X_{n-1}$ denote the number of required applications of $Q^{\Simon}_{h \circ f}$ in \textsc{Hashed-Simon}. Then
		\begin{align*}
			\E[X]  	& = \sum_{i=1}^{n-1} \E[X_i] = \sum_{i=1}^{n-1}\frac{1}{(1-2^{-t})}\cdot  \frac{2^{n-1}}{2^{n-1}-2^{i-1}} 
			= \sum_{i=1}^{n-1} \frac{1}{(1-2^{-t})}\cdot  \frac{2^{n-1}-2^{i-1}+2^{i-1}}{2^{n-1}-2^{i-1}}  \\
			& = \frac{1}{(1-2^{-t})}\cdot (n-1) + \frac{1}{(1-2^{-t})}\cdot \sum_{i=1}^{n-1} \frac{2^{i-1}}{2^{n-1}-2^{i-1}}.
		\end{align*}
		Since $\lim_{n \to \infty} \sum_{i=1}^{n-1} \frac{2^{i-1}}{2^{n-1}-2^{i-1}} \leq 1.6067$, the claim follows.
	\end{proof}	

\begin{remark}
	With a similar analysis as in the proof of \cref{main_theorem_simon}, we obtain an upper bound of $n + 1$ for the expected number of applications of $Q^{\Simon}_{f}$ in \textsc{Simon}'s original algorithm.
\end{remark}

%
%

\section{Simon Attack on the Even-Mansour Construction}\label{sec:even_mansou}

The famous Even-Mansour construction~\cite{DBLP:journals/joc/EvenM97,DBLP:conf/eurocrypt/DunkelmanKS12} is an appealingly simple way of constructing a keyed pseudo-random permutation from an unkeyed public permutation $P:\F_2^n\to\F_2^n$ via 
\[
	\EM_{k}: \F_2^n \rightarrow \F_2^n, \quad x \mapsto P( x + k ) + k\;.
\] 
As shown by Even, Mansour~\cite{DBLP:journals/joc/EvenM97} and Dunkelman et al~\cite{DBLP:conf/eurocrypt/DunkelmanKS12}, the function $\EM_{k}$ offers quite strong security guarantees against classical adversaries.

However, Kuwakado and Morii \cite{DBLP:conf/isita/KuwakadoM12} and Kaplan et al~\cite{DBLP:conf/crypto/KaplanLLN16} showed that Even-Mansour is completely insecure against quantum superposition attacks. The key observation is that the function
\[ f: \F_2^n \rightarrow \F_2^n, \quad x \mapsto P(x) + \EM_{k}(x) = P(x) + P( x + k )  + k \]
satisfies $f(x) = f(x+k)$. Thus, an application of Simon's algorithm reveals as period the secret key $k$. This requires $2n$ qubits using \textsc{Simon}. Let $U_P$ and $U_{\EM_{k}}$ be quantum embeddings of $P$ and $\EM_{k}$. The quantum circuit for the attack is depicted in \cref{circuit:EM}.

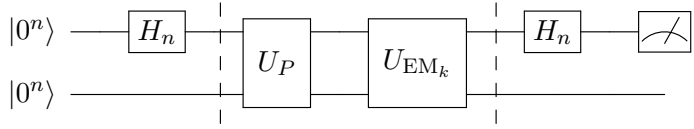
\begin{figure}[htb]
	\centering
	\mbox{
		\Qcircuit @C=1em @R=1em {
			\lstick{\ket{0^n}}		& 
			\qw 				 	& \gate{H_n}& \qw\barrier[-0.75cm]{1}  	& \multigate{1}{U_{P}} 	& \qw & \multigate{1}{U_{\EM_{k}}} & \qw\barrier[-0.75cm]{1} & \gate{H_n} 	& \qw 					 & \meter\\
			\lstick{\ket{0^n}}        & 
			\qw 					& \qw 		& \qw 						& \ghost{U_{P}}	& \qw & \ghost{U_{\EM_{k}}}& \qw 				   	 & \qw 			& \qw					 & \qw 
		}
	}
	\caption{Quantum circuit for \textsc{Simon}-attack on Even-Mansour.}\label{circuit:EM}
\end{figure}

\subsection{Directly Realizing Hashed Even-Mansour: \textsc{Simon} Attack with $n+1$ qubits  }

From \cref{main_theorem_simon} we immediately conclude that we obtain a \textsc{Simon}-attack with $n+1$ qubits using oracle access to hashed versions of $f$. In the following we show that we can directly (without oracles) construct hashed versions from $P$ and $\EM_{k}$. 

First observe that $P$ and $\EM_k$ are reversible functions $\F_2^n \rightarrow \F_2^n$, and thus may allow for direct quantum circuits that compute the function values of $P$, $\EM_k$ on the $n$ input qubits, without using the generic universal quantum embedding strategy. Let us assume for the moment that we are able to construct for $P$ an  {\em in-place} quantum circuit $Q_P$, i.e. a circuit that acts on the $n$ input qubits only. We show in the following that some natural choices for $P$ in the Even-Mansour construction allow for such in-place realizations. 

If $P$ can be realized via $Q_P$ in-place, then we can also realize $\EM_k$ in-place via some circuit $Q_{\EM_k}$, where $Q_{EM_k}$ just uses $Q_P$ and adds in the key $k$ (hardwired). Running the implementations of $Q_P$, $Q_{\EM_k}$ backwards realizes the inverse functions $P^{-1}, \EM_k^{-1}$. We denote the corresponding circuits by $Q_P^{-1}$ and $Q_{\EM_k}^{-1}$. Notice that our hash method comes at the cost of doubling the circuit depth.

We take the following universal hash family from \cref{eq:hash} with $t=1$
\[ 
\mathcal{H}_1=\{h_r:\F_2^n\to\F_2 \;|\; r\in \F_2^n,\ h_r(x)=\langle x , r \rangle \}.
\] 
Notice that $\mathcal{H}_1$ is homomorphic, i.e. for all $h_r \in \mathcal{H}_1$ we have $h_r(x) + h_r(y) = h_r(x+y)$ by linearity of the inner product. As usual, we denote by $U_{h_r}$ the universal embedding of $h_r$.

The $n+1$ qubit quantum circuit $Q_{HS}$ in \cref{circuit:HEM} describes a \textsc{Hashed-Simon} attack on Even-Mansour {\em without} the need for oracle access to hashed versions for $f$.

\begin{figure}[htb]
	\centering
	\includegraphics[width=\textwidth]{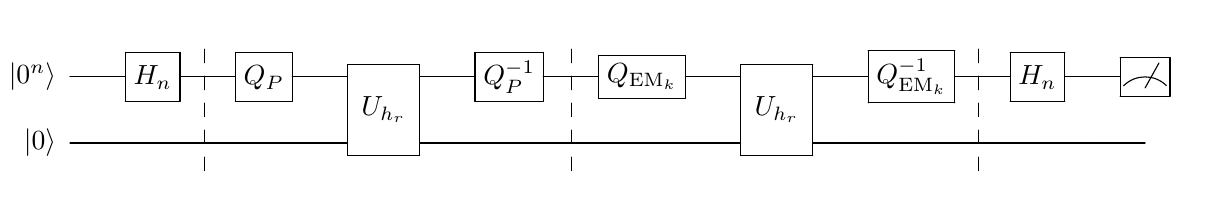}
	\caption{Quantum circuit $Q_{\textsc{HS}}$ for a \textsc{Hashed-Simon}-Attack on Even-Mansour}\label{circuit:HEM}
\end{figure}

\noindent In \cref{circuit:HEM} we compute on the single output qubit of $Q_{\textsc{HS}}$ 
\begin{align*}
h_r(P(x)) + h_r(\EM_k(x)) & = \langle P(x),r\rangle + \langle \EM_k(x),r\rangle = \langle P(x) + \EM_k(x) ,r\rangle  \\
 & = h_r(P(x) + \EM_k(x)) = h_r(f(x)).
\end{align*}
Thus, the correctness of our construction follows. 

It remains to show that we can compute $P$ in-place.

\paragraph*{In-place realization of $P$.}

There exist many lightweight permutations such as Gimli \cite{DBLP:conf/ches/BernsteinKLMMN017} that allow for (quantum) hardware-efficient implementations, for a list of candidates see the current second round NIST competition~\cite{NIST} or the work of Bonnetain and Jaques~\cite{DBLP:journals/corr/abs-2011-07022}. 
For didactical reasons -- since it is especially easy to describe and implement in-place -- we choose the SiMeck cipher~\cite{DBLP:conf/ches/YangZSAG15}, for  which we fix the key $k' \in \F_2^{n/2}$ to obtain a public permutation $P$, as also done in ACE \cite{ACE}. 

SiMeck is a round-iterated Feistel cipher, see \cref{fig:feistel}
for one Feistel round ${\cal F}_{k'}: \F_2^n \rightarrow \F_2^n$. 

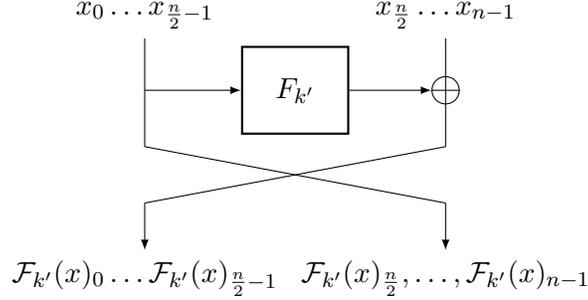
\begin{figure}[htb]
	\centering
	\begin{tikzpicture}[auto,transform shape]
		\tikzstyle{block} = [rectangle, draw, thick,
		text width=3em, text centered, minimum height=3em]
		\tikzstyle{line} = [draw, -]
	
		\node [block] (Et) {$F_{k'}$};
		\coordinate [above of=Et, node distance=1cm] (top);

		\node [left of=top, node distance=2cm] (m1) {$x_0 \ldots x_{\frac{n}{2}-1}$};
		\node [right of=top, node distance=2cm] (m2) {$x_\frac{n}{2} \ldots x_{n-1}$};

		\coordinate [left of=Et, node distance=2cm] (midleft);
		\node [right of=Et, node distance=2cm, draw, circle] (xor) {};
		
		\draw[-] (xor.north) -- (xor.south);
		\draw[-] (xor.east) -- (xor.west);
		
		\coordinate [left  of=xor, node distance=1mm] (xorl);
		\coordinate [above of=xor, node distance=1mm] (xora);
		\coordinate [right of=xor, node distance=1mm] (xorr);
		\coordinate [below of=xor, node distance=1mm] (xorb);
	
		\coordinate [below of=midleft, node distance=.75cm] (botleft);
		\coordinate [below of=botleft, node distance=.75cm] (bumleft);
		\coordinate [right of=botleft, node distance=4cm] (botright);
		\coordinate [right of=bumleft, node distance=4cm] (bumright);
		
		\node [below of=m1, node distance=3.5cm] (c1) {${\cal F}_{k'}(x)_0 \ldots {\cal F}_{k'}(x)_{\frac{n}{2}-1}$};
		\node [right of=c1, node distance=4cm] (c2) {${\cal F}_{k'}(x)_\frac{n}{2},\ldots,{\cal F}_{k'}(x)_{n-1}$};
	
		\path [line] (m1) -- (midleft);
		\path [line] (midleft) -- (botleft);
		\path [line] (botleft) -- (bumright);
		\draw[-latex] (bumright) -- (c2);
	
		\path [line] (m2) -- (xor.north);
		\path [line] (xor.south) -- (botright);
		\path [line] (botright) -- (bumleft);
		\draw[-latex] (bumleft) -- (c1);
		
		\draw[-latex] (midleft) -- (Et);
		\draw[-latex] (Et) -- (xor.west);	
	\end{tikzpicture}
	\caption{One round of ${\cal F}_{k'}$ of SiMeck, where $F_{k'}: \F_2^{\frac n 2} \rightarrow \F_2^{\frac n 2}$ is the round function.}\label{fig:feistel}
\end{figure}

\noindent Let $x = x_0 \ldots x_{n-1}$. Then the $i^{th}$ bit of ${\cal F}_{k'}(x)$ is
\[
  {\cal F}_{k'}(x)_i = 
\begin{cases}  
x_{i-n/2} & \textrm{if } i \geq n/2 \\
F_{k'}(x_0, \ldots, x_{n/2-1})_i + x_{i + n/2} & \textrm{if } i < n/2
\end{cases}.
\]
The round function $F_{k'}: \F_2^{n/2} \rightarrow \F_2^{n/2}$ in SiMeck is defined as  
\[
F_{k'}(x)_i = x_{i}\cdot x_{i+5 \bmod n/2} + x_{i+1 \bmod n/2} + k'_i \ .
\]
We implement ${\cal F}_{k'}$ in-place as depicted in \cref{fig:simeck_cipher_iteration}. In our quantum circuit the AND-operation is realized by a Toffoli gate. Conditioned on $k'_i=1$ we place a NOT-gate $X$ (i.e. we hardwire $k_i'$).

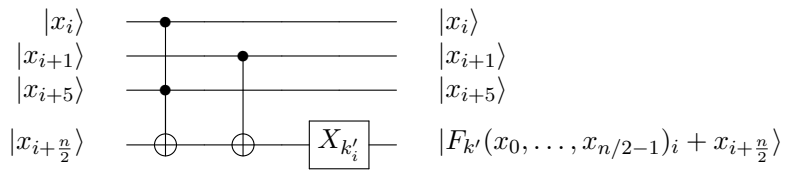
\begin{figure}[htb]
	\centering
	\scalebox{0.95}{
		\Qcircuit @C=1em @R=1em {
			\lstick{\ket{x_{i}}}					&& \ctrl{3}	& \qw	& \qw		& \qw	& \qw 			& \qw 	& \rstick{\ket{x_{i}}}\\		
			\lstick{\ket{x_{i+1}}}					&& \qw		& \qw	& \ctrl{2}	& \qw	& \qw 			& \qw	& \rstick{\ket{x_{i+1}}}\\
			\lstick{\ket{x_{i+5}}}					&& \ctrl{1}	& \qw	& \qw		& \qw	& \qw 			& \qw	& \rstick{\ket{x_{i+5}}}\\	
			\lstick{\ket{x_{i+\frac{n}{2}}}}		&& \targ	& \qw	& \targ	 	& \qw	& \gate{X_{k'_i}}& \qw	& \rstick{\ket{F_{k'}(x_0, \ldots, x_{n/2-1})_i + x_{i+\frac{n}{2}}}}
		}
	}
	\caption{In-place realization of ${\cal F}_{k'}(x)_i$, $i<\frac{n}{2}$. If $k'_i=1$ we place a NOT-gate $X$.}\label{fig:simeck_cipher_iteration}
\end{figure}

Going from the public permutation $P$ realized via SiMeck to Even-Mansour $\EM_{k}$, we also hardwire the bits of $k$ via NOT-gates, see \cref{fig:even_mansou}.

\begin{figure}[htb]
	\centering
	\mbox{
		\Qcircuit @C=1em @R=1em {
			\lstick{\ket{x_0}} 		 						&& \qw  & \gate{X_{k_0}}				& \multigate{2}{Q_P}	& \gate{X_{k_0}}			& \qw 	& \rstick{\ket{\EM(x)_0}}\\
			\lstick{\raisebox{.2cm}{\vdots}\hspace{0.5cm}} 	&&   	& \raisebox{.2cm}{\vdots}		&  						& \raisebox{.2cm}{\vdots}	&		& \rstick{\hspace{0.7cm}\raisebox{.2cm}{\vdots}}	\\
			\lstick{\ket{x_{n-1}}}	 						&& \qw  & \gate{X_{k_{n-1}}}			& \ghost{Q_P}		 	& \gate{X_{k_{n-1}}}		& \qw 	& \rstick{\ket{\EM(x)_{n-1}}}
		}
	}
	\caption{Quantum circuit for Even-Mansour. We place a NOT-gate $X$ if $k_i=1$.}
	\label{fig:even_mansou}
\end{figure}
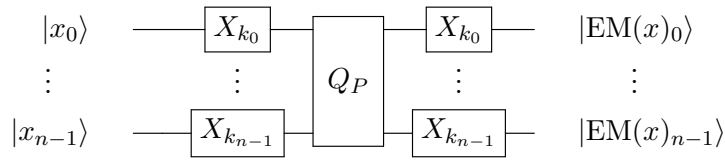
\section{Hashing Offline Even-Mansour to a Quarter of its Bits}\label{sec:offline_simon}

While Simon's attack on the Even-Mansour cipher 
\[
\EM_k: \F_2^n \rightarrow \F_2^n, \ x \mapsto P(x+k)+k 
\] 
from \cref{sec:even_mansou} with only $\bigO(n)$ queries nicely illustrates the power of quantum computations on symmetric cryptography, it also uses a strong model giving an attacker full {\em quantum access} to $\EM_k$. 

Recently, Bonnetain et al~\cite{DBLP:conf/asiacrypt/BonnetainHNSS19} proposed an algorithm called \textsc{Offline-Simon} that in a more realistic  model, where an attacker gets only {\em classical access} to $\EM_k$, achieves a {\em polynomial} speedup over classical attacks. More precisely, the Even-Mansour attack with \textsc{Offline-Simon} runs in time $\tilde{\bigO}\left(2^{\frac{n}{3}}\right)$ using $\frac{4}{3}cn^2 + o(n^2)$ qubits, for some constant $c$ (chosen as $c=\frac 5 3$ in~\cite{DBLP:conf/asiacrypt/BonnetainHNSS19}). This qubit analysis uses the very mild assumption that Even-Mansour's public permutation $P: \F_2^n \rightarrow \F_2^n$ can be implemented with $o(n^2)$ qubits. Efficiently computable $P$'s are usually computable with $\bigO(n)$ qubits.  

We show in the following that our hashing technique reduces the number of required qubits to only $\frac{1}{3}cn^2 + o(n^2)$, thereby saving roughly a factor of $4$. 

For simplicity of exposition we explain the \textsc{Offline-Simon} technique only when applied to Even-Mansour. Moreover, we ignore  the fact that Even-Mansour is not a perfect 2:1-function, which only insignificantly affects the analysis as shown in \cite{DBLP:conf/crypto/KaplanLLN16,DBLP:journals/qic/SantoliS17,DBLP:conf/asiacrypt/Leander017}.

In general, \textsc{Offline-Simon} has more applications such as the FX-construction attack~\cite{DBLP:conf/asiacrypt/Leander017}, and our hashing technique transfers to these applications as well. 
For instance for the FX-construction we save via hashing a factor of $2$ in the number of qubits. 

\paragraph*{Offline-Simon.} 

Recall from \cref{sec:even_mansou} that the main observation of the quantum attack on Even-Mansour is that the function $\EM_k(x) + P(x)$ is periodic with period $k$. In this function only $\EM_k$ is key-dependent. The idea of \textsc{Offline-Simon} is to define a key-dependent function
\[
 g:\F_2^{\frac{n}{3}}\to\F_2^n , \ x \mapsto \EM_k(x||0^{\frac{2n}{3}}), 
\]
for which we only have classical access, and a key-independent function with quantum access
\[
 f_{k'}:\F_2^{\frac{n}{3}}\to\F_2^n, \ x \mapsto P(x||k') \textrm{ for some } k'\in\F_2^\frac{2n}{3}.
\]
We write the secret key as $k=(k_1 || k_2)\in\F_2^{\frac{n}{3}}\times\F_2^{\frac{2n}{3}}$. Notice that 
\begin{align*}
g(x) + f_{k_2}(x) & =  \EM_{(k_1||k_2)}(x||0^{\frac{2n}{3}}) + P(x || k_2)  \\
& = P(x + k_1||k_2) + k + P(x || k_2)
\end{align*}
is periodic in $k_1 \in \F_2^{\frac n 3}$. As opposed to \cref{sec:even_mansou} instead of period $k$ we obtain period $k_1$, thereby reducing the period length from $n$ to $\frac n 3$ bits. However, $g + f_{k'}$ is periodic only for the choice $k'=k_2$, otherwise it behaves like a random function (by the property of $P$). Therefore, \textsc{Offline-Simon}  searches for $k'=k_2$ with a Grover search in time $\tilde{\bigO}(2^{\frac n 3})$,  using a  Grover function
that returns $1$ iff $g + f_{k'}$ is periodic. 
Such a function has been designed by Leander and May~\cite{DBLP:conf/asiacrypt/Leander017}.

In a nutshell, \textsc{Offline-Simon} proceeds as follows. Let $c$ be a small constant. Using $2^{n/3}$ classical queries to $g$ we determine $g(x)$ for all $x \in \F_2^{\frac n 3}$. We then build the quantum state 
\[ \ket{\Phi_g}:=\bigotimes_{j=1}^{cn}\left(\sum_{x_j \in \F_2^{\frac n 3}}\ket{x_j}\ket{g(x_j)}\right) .
\]
Using $\bigO(n)$ quantum queries to $f_{k'}$ we construct the superposition
\[
\ket{\Phi_{g+f_{k'}}}:=\bigotimes_{j=1}^{cn}\left(\sum_{x_j \in \F_2^{\frac n 3}}\ket{x_j}\ket{(g+f_{k'})(x_j)}\right).
\]
Eventually, we use Hadamard on the $\ket{x_j}$'s to obtain $cn$ copies of a typical \textsc{Simon} superposition
\begin{equation}\label{eq:phi}
\ket{\Phi} = \left(\sum_{x_1,y_1 \in \F_2^{\frac n 3}}(-1)^{x_1y_1}\ket{y_1}\ket{(g+f_{k'})(x_1)}\right) \otimes \ldots \otimes \left(\sum_{x_{cn},y_{cn}\in \F_2^{\frac n 3}}(-1)^{x_{cn}y_{cn}}\ket{y_{cn}}\ket{(g+f_{k'})(x_{cn})}\right)\;.
\end{equation}

The state $\ket{\Phi}$ is checked for periodicity of $g+f_{k'}$. Notice that $\ket{\Phi}$ has $cn$ copies of a state with $\frac n 3$ input qubits for $x$ and $n$ output qubits for $(g+f_{k'})(x)$. Thus ignoring low order terms we need $\frac 4 3 cn^2$ qubits.

\paragraph*{Hashed-Offline-Simon.} Let us use our $t$-output bit homomorphic hash function family from \cref{eq:hash}
\[
\mathcal{H}_t=\{h_r:\F_2^n\to\F_2^t \;|\; r\in \left(\F_2^n\right)^t,\ h_r(x)=(\langle x , r_1 \rangle,\ldots,\langle x , r_t \rangle) \}\;. 
\]
Let $c'$ be a small constant. In the following we see how $c'$ relates to $c$ from \textsc{Offline-Simon}.

Choose $h_j \in_R \mathcal{H}_t$ for $1 \leq j \leq c'n$. As in \textsc{Offline-Simon}, we first classically query on-the-fly for all $x_j $  the value $g(x_j )$ and compute $h_j(g(x_j ))$ for every $j$ to create
%
\[ \ket{\Phi_{h \circ g}}:=\bigotimes_{j=1}^{c'n}\left(\sum_{x_j \in \F_2^{\frac n 3}}\ket{x_j}\ket{h_j(g(x_j ))}\right) .
\]
Second, using $\bigO(n)$ quantum queries to $h_j \circ f_{k'}(x_j ) = h_j(P(x_j ||k'))$ (a combination of each $h_j$ with the quantum circuit for $f_{k'}$) we construct\footnote{
This can be done by computing iteratively $f_{k'}(x_j)$, hashing it with $h_j$, and uncomputing $f_{k'}(x_j)$ to reuse the qubits 
for $f_{k'}(x_{j+1})$. Using this iterative procedure we only need once instead of $n$ times the qubits for representing $f_{k'}$.
} 
the superposition
\[
\ket{\Phi_{h\circ (g+ f_{k'})}}:=\bigotimes_{j=1}^{c'n}\left(\sum_{x_j  \in \F_2^{\frac n 3}}\ket{x_j }\ket{(h_j \circ g+ h_j \circ f_{k'})(x_j )}\right).
\]
By the homomorphic property of $h_j \in \mathcal{H}_t$ we have $h_j \circ g+ h_j \circ f_{k'} = h_j \circ (g+f_{k'})$. 

Eventually, Hadamard on $\ket{x}$ creates $c'n$ copies of a \textsc{Simon} superposition
\begin{equation}\label{eq:phih}
	\begin{split}
		\ket{\Phi_h} =& \left(\sum_{x_1,y_1\in \F_2^{\frac n 3}}(-1)^{x_1y_1}\ket{y_1}\ket{h_1 \circ (g+f_{k'})(x_1)}\right)\otimes \ldots  \\
		&\otimes \left(\sum_{x_{c'n},y_{c'n}\in \F_2^{\frac n 3}}(-1)^{x_{c'n}y_{c'n}}\ket{y_{c'n}}\ket{h_{c'n} \circ (g+f_{k'})(x_{c'n})}\right)\;.
	\end{split}
\end{equation}

Thus, in contrast to \textsc{Offline-Simon} we hash the $n$ output qubits to $t$ output qubits. 
However, we have to chose $t$ with some care. E.g. taking $t=1$ minimizes the number of bits per copy from $\frac 4 3 n$ to only $\frac 1 3 n + 1$. However, 
by \cref{measuring_probabilities} the choice $t=1$ results in (roughly) half the $y_i$'s being zero, which in turn forces us to set $c'\ge2c$. Thus, in total we obtain at least $2cn (\frac 1 3 n + 1)$ qubits instead of $cn(\frac 4 3 n)$, saving at most a factor of (roughly) $2$. We will show in  \cref{theo:hashEM} that the choice $t = \log_2 n$ saves us a factor of (roughly) $4$ by reducing the bits per copy to (the achievable minimum) $\frac 1 3 n+o(n)$ while increasing $cn$ insignificantly to only $c'n=cn + o(n)$.

\begin{theorem}
\label{theo:hashEM}
    A \textsc{Hashed-Offline-Simon} attack with only classical access to the Even Mansour cipher $\EM_k$ computes $k$ in time $\tilde{\bigO}\left(2^{\frac{n}{3}}\right)$ using $\frac 5 {9} {n^2} + o(n^2)$ qubits, instead of $\frac{20}{9} n^2 + o(n^2)$ qubits for the non-hashed version.
\end{theorem}

\begin{proof}
	We know from~\cite{DBLP:conf/asiacrypt/BonnetainHNSS19} and the above discussion that the non-hashed version \textsc{Offline-Simon} works for the choice $c = \frac 5 3$ in time  $\tilde{\bigO}\left(2^{\frac{n}{3}}\right)$ using $\frac{4}{3}cn^2 + o(n^2)=\frac{20}{9}n^2 + o(n^2)$ qubits, where $o(n^2)$ also accounts for the required ancilla qubits.
	
	We choose $t=\log_2(n)$ for the hash length of our family $\mathcal{H}_t$.
	We measure each $y_j$ in $\ket{\Phi}$ from \cref{eq:phi} with probability $\frac{1}{2^{n/3-1}}$.
	By \cref{measuring_probabilities}, we measure each $y_j \not=0$ in $\ket{\Phi_h}$ from \cref{eq:phih} with probability $\frac{1-2^{-t}}{2^{n/3-1}}$.
	This implies that conditioned on measuring $y_j\not= 0$, the probability distribution of all $y_j$ is preserved.
	
	Let $X_i$ be a Bernoulli random variable that takes value $1$ iff we measure $y_i=0$.  We have $p:=\mathbb{P}[X_i=1] = \frac{1}{n}+\frac{n-1}{n\cdot 2^{n/3-1}} <\frac{2}{n}$ for sufficiently large $n$.
	Let 
	\[ m = \frac {n}{2 \ln(n)}. 
	\]
	We choose $c'n=cn + m$ in \textsc{Hashed-Offline-Simon}. 
	
	Let $\texttt{BAD}$ be the  event  that less than $cn$ of our $c'n$ measurements are non-zero vectors $y_i\not = 0$. Notice that in the event $\texttt{BAD}$ we do not have sufficiently many non-zero vectors for \textsc{Hashed-Offline-Simon}. Let $X = \sum_{i=1}^{c'n} X_i$ be a random variable for the number of $0$-measurements. Then $\mu := \E[X] = pc'n < 2c + \frac{2m}{n} < 4$ for sufficiently large $n$.

	Application of a Chernoff bound yields
	\begin{align*}
		\mathbb{P}[\texttt{BAD}]
		&=\mathbb{P}\left[X \geq c'n-cn\right]
		= \mathbb{P}\left[X \ge m \right] = \mathbb{P}\left[X \ge \left(1 + \left(\frac m {\mu } -1 \right) \right) \cdot \mu \right]  \\
		& \leq \left( \frac{e^{m/\mu -1}}{(m/\mu)^{m/\mu}}\right)^{\mu}  = e^{m - \mu - \ln(m/\mu) \cdot m }  = e^{-n/2 + o(n)} < 2^{-n/3}\;,
	\end{align*}
	for sufficiently large $n$.
	
	\textsc{Hashed-Offline-Simon} requires a total of $\bigO\left(2^{n/3}\right)$ iterations. Thus, we upper bound the iterations by $\ell \cdot 2^{n/3}$ for some constant $\ell$.
	For sufficiently large $n$, we may lower bound the probability to obtain in each of the iteration at least $cn$ vectors $y_j\not=0$ by
	\begin{align*}
		\mathbb{P}[\overline{\texttt{BAD}}]^{\ell \cdot 2^{n/3}}
		=\left(1-\mathbb{P}[\texttt{BAD}]\right)^{\ell \cdot 2^{n/3}} = \left( (1-2^{-n/3})^{2^{n/3}} \right)^\ell
		\ge\left( \frac 1 2 e^{-1} \right)^{\ell} = \Omega(1).
	\end{align*}
	In $\ket{\Phi_h}$ we obtain $c'n$ copies of $n/3 + t$ and therefore a total qubit amount of 
	\[
	c'n \cdot \left(\frac n 3 + \log_2 n\right) 
	=\left(cn + o(n)\right)\left(\frac n 3 + o(n)\right)
	= \frac{5}{9} n^2 + o(n^2).
	\]
\end{proof}

\medskip

%
%

\section{Hashed Shor: Special Periods}
\label{hashed_shor}


Let us briefly recall Shor's algorithm.  Let  $f: \Z \rightarrow \Z$ be periodic with period $d \in \N$, i.e. $d>0$ is minimal with the property $f(x) = f(x+d)$ for all $x \in \Z$.  For ease of notation, let us first focus on applying Shor's algorithm for factorization. In \cref{ekera_sec} we will also see an application for discrete logarithms. 

Let $N \in \N$ be a composite $n$-bit number of unknown factorization, and let $a$ be chosen uniformly at random from $\Z_N^*$, the multiplicative group modulo $N$. Let us define the function  
$f: \Z \rightarrow \Z_N$, $x \mapsto a^x \bmod N$. Notice that $f$ is periodic with $d=\textrm{ord}_N(a)$, since $f(x+d) = a^{x+d} = a^x a^{\textrm{ord}(a)} = a^x = f(x)$. It is well-known that we can compute a non-trivial factor of $N$ in probabilistic polynomial time given $d=\textrm{ord}_N(a)$~\cite{DBLP:journals/siamcomp/Shor97}. We encode the inputs of $f$ with $q$ qubits.

In order to find $d$, Shor uses the quantum circuit $Q^{\Shor}_{f}$ from \cref{circuit:shor} with oracle-access to $f$. In $Q^{\Shor}_{f}$ we measure in the $q$ input qubits with high probability $y$'s that are close to some multiple of $\frac{2^q}{d}$. The original \textsc{Shor} algorithm then measures sufficiently many $y$'s (a constant number is sufficient) to extract $d$ in a classical post-process.

\begin{figure}[h]
	\centering
	\mbox{
		\Qcircuit @C=1em @R=1em {
			\lstick{\ket{0^{q}}}	& \qw & \gate{H_{q}} 	& \qw & \multigate{1}{U_{f}} & \qw & \gate{\text{QFT}_{q}} 	& \qw & \meter\\
			\lstick{\ket{0^n}}		& \qw & \qw	 	   		& \qw & \ghost{U_{f}}        & \qw & \qw	     			& \qw & \qw
		}
	}
	\caption{Quantum circuit $Q^{\Shor}_{f}$ 
		}\label{circuit:shor}
\end{figure}
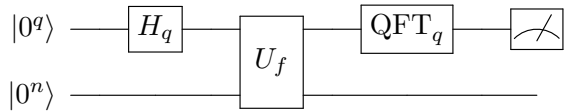

Our \textsc{Hashed-Shor} (Algorithm~\ref{alg:shor}) simply replaces circuit $Q_f^{\textsc{Shor}}$ with its hashed version $Q_{h \circ f}^{\textsc{Shor}}$, where we use oracle-access to hashed versions of $f$. Notice that \textsc{Shor} is a special case of \textsc{Hashed-Shor} for the choice $t=\lceil\log_2(N)\rceil$ and $\mathcal{H}_t=\{\id\}$. For this choice $\mathcal{H}_t$ is not universal, but we do not need universality in the following \cref{Hashed-Shor state} about the superposition produced by $Q_{h \circ f}$. From \cref{Hashed-Shor state} we conclude correctness of \textsc{Hashed-Shor} for any $t$-bit range hash function $h$.

\begin{algorithm}[ht]
	\DontPrintSemicolon
	\SetAlgoLined
	\SetKwInOut{Input}{Input}\SetKwInOut{Output}{Output}
	
	\SetKwComment{COMMENT}{$\triangleright$\ }{}
	
	\Input{ $f: \Z \rightarrow \Z_N$, universal $\mathcal{H}_t:= \{h: \Z_N \rightarrow \HD^t\}$}
	\Output{ Period $d$ of $f$}
	
	\Begin{
		Set $Y=\emptyset$.\; 
		\Repeat{$|Y|$ is sufficiently large.}{
			Run $Q^{\textsc{Shor}}_{h \circ f}$ on $\ket{0^q}\ket{0^t}$ for some freshly chosen $h \in_R \mathcal{H}_t$. \;
			Let $y$ be the measurement of the $q$ input qubits. \;
			If $y \not= 0$, then include $y$ in $Y$.
		}
		Compute $d$ from $Y$ in a classical post-process.\;
		\KwRet{$d$}
	}
	\caption{\textsc{Hashed-Shor}\label{alg:shor}}
\end{algorithm}

\newpage
\begin{lemma}\label{Hashed-Shor state}
	Let $N \in \mathbb{N}$, $a \in \Z_N^*$ with $d=\textrm{ord}_N(a)$ and $f(x)=a^x \bmod N$. Let $h: \Z_N \rightarrow \HD^t$. Define $M_z := \{k \in \Z_d \mid h(a^k \bmod N) = z\}$. An application of quantum circuit $Q^{\Shor}_{h \circ f}$ on input $\ket{0^q}\ket{0^t}$ yields a superposition
	\begin{equation}
	\label{hshor_eqn}
	\ket{\Phi_h}=\sum_{y=0}^{2^q-1}\sum_{z\in\HD^t}\frac{1}{2^q}\sum_{k\in M_z}\sum_{{c \geq 0:} \atop cd+k < 2^q} e^{2\pi i \frac{cd+k}{2^q}y}\ket{y}\ket{z}\;.
	\end{equation}
\end{lemma}

\begin{proof}
	In $Q^{\Shor}_{h \circ f}$, we apply on input $\ket{0^q}\ket{0^t}$ first the operation $H_q \otimes I_t$ followed by $U_{h\circ f}$. This results in superposition
	\[\frac{1}{\sqrt{2^q}}\sum_{x=0}^{2^q-1}\ket{x}\ket{h(a^x\bmod N)} .\]
	Let $x = cd + k$ with $k \in \Z_d$. Since $a^{x} = a^{cd + k} \equiv a^k \mod N$, the value of $f(x)$ depends only on $k = (x \bmod d)$. Therefore, we rewrite the above superposition as
	\[
	\frac{1}{\sqrt{2^q}}\sum_{z\in\HD^t}\sum_{k\in M_z}\sum_{{c \geq 0:} \atop cd+k < 2^q}\ket{cd+k}\ket{z}\;,
	\]
	Eventually, an application of \QFT$_{q}$ yields
	\[
	\ket{\Phi_h}=\sum_{y=0}^{2^q-1}\sum_{z\in\HD^t}\underbrace{\frac{1}{2^q}\sum_{k\in M_z}\sum_{c \geq 0: \atop cd+k < 2^q} e^{2\pi i \frac{cd+k}{2^q}y}}_{w_{y,z}}\ket{y}\ket{z}\;.\]
\end{proof}

\begin{remark}
	For the choice $h=id$, \cref{Hashed-Shor state} provides an analysis of Shor's original quantum circuit $Q^{\Shor}_{f}$. This choice implies $M_z  = \{k \in \Z_d \mid a^k = z \text{ mod } N \}$. Therefore, we obtain the superposition 
	\begin{equation}
	\label{shor_eqn}
	\ket{\Phi}=
	\sum_{y=0}^{2^q-1}\sum_{k=0}^{d-1} \frac{1}{2^q}\sum_{c \geq 0: \atop cd+k < 2^q}e^{2 \pi i \frac{cd+k}{2^q}y}\ket{y}\ket{a^k \text{ mod } N}\;.
	\end{equation}
	and the amplitudes of $\ket{y}\ket{z_k}$ with $z_k=a^k \text{ mod } N$ are
	\begin{equation*}
	w_{y,z_k}=\frac{1}{2^q} \sum_{c \geq 0: \atop cd+k < 2^q}e^{2 \pi i \frac{cd+k}{2^q}y}\;.
	\end{equation*}
\end{remark}

For didactical reasons and ease of notation, let us look in the subsequent section at the special case of periods $d$ that are powers of two. In Section~\ref{shor_general}, we analyse the general $d$ case.

\subsection{Periods that are a power of two}

Let $d=2^r$ for some $r\in\mathbb{N}$ with $r \leq q$. Then $\max\{ c \in \mathbb N \mid cd+k < 2^q\} = \frac{2^q}{d}-1 =2^{q-r}-1$, independent of $k \in \Z_d$. Hence, let us define $m=\frac{2^{q}}{d}$ and $z_k= a^k\mod N$. Using $md = 2^q$, 
this allows us to rewrite Eq.~(\ref{shor_eqn}) and Eq.~(\ref{hshor_eqn}) as
\begin{equation}
\label{shor_eqn2}
\ket{\Phi}=
\sum_{y=0}^{2^q-1} \sum_{k=0}^{d-1}\left(\frac{1}{\sqrt{d}}\cdot e^{2 \pi i \frac{k}{2^q}y}\vphantom{\frac{1}{\sqrt{m\cdot2^q}}\sum_{c=0}^{m-1}e^{2 \pi i \frac{cd}{2^q}y}}\right)\cdot\left(\frac{1}{\sqrt{m\cdot2^q}}\sum_{c=0}^{m-1}e^{2 \pi i \frac{cd}{2^q}y}\right)\ket{y}\ket{z_k}
\end{equation}
respectively for $h: \Z_N \rightarrow \HD^t$ as
\begin{equation}
\label{hshor_eqn2}
\ket{\Phi_h}=
\sum_{y=0}^{2^q-1}
\sum_{z\in\HD^t}\left( \frac{1}{\sqrt{d}}  \sum_{k\in M_z} e^{2\pi i \frac{k}{2^q}y}\right)\cdot\left(\frac{1}{\sqrt{m\cdot2^q}}\sum_{c=0}^{m-1}e^{2 \pi i \frac{cd}{2^q}y}\right)\ket{y}\ket{z}
\;.
\end{equation}
Notice that the factor 
\[
\frac{1}{\sqrt{m\cdot2^q}}\sum_{c=0}^{m-1}e^{2 \pi i \frac{cd}{2^q}y}
\]
is identical in $\ket{\Phi}$ and its hashed version $\ket{\Phi_h}$. Further notice that the factor is independent of $z_k$ and $z$. In the following lemma we show that for a measurement of any $\ket{y}$, where $y$ is a multiple of $m$, this factor contributes to the probability with $\frac 1 d$.
\begin{lemma}
	\label{multiple}
	Let $d=2^r \leq 2^q$ and $y = \ell m$ for some $0 \leq \ell < d$. 
	Then we have
	\[
	\left| \frac{1}{\sqrt{m\cdot2^q}}\sum_{c=0}^{m-1}e^{2 \pi i \frac{cd}{2^q}y} \right|^2 = \frac 1 d\;.
	\]
\end{lemma}

\begin{proof}
	Since $y = \ell m = \ell \frac{2^q}{d}$ we obtain 
	\[
	\left|\frac{1}{\sqrt{m\cdot2^q}}\sum_{c=0}^{m-1} e^{2\pi i \frac{cd}{2^q}y}\right|^2
	=
	\frac{1}{m\cdot2^q}\left|\sum_{c=0}^{m-1} e^{2\pi i c \ell}\right|^2
	= \frac{m^2}{m\cdot2^q}
	=
	\frac{m}{2^q}
	=
	\frac{1}{d}\;.\vspace*{-0.5cm}
	\]
\end{proof}

We now show that the same common factor ensures that in both superpositions $\ket{\Phi}$ and its hashed version $\ket{\Phi_h}$ we never measure some $\ket{y}$ if $y$ is not a multiple of $m$. 

\begin{lemma}
	\label{non-multiple}
	Let $d=2^r \leq 2^q$ and $y \in \{0, \ldots, 2^q-1\}$ with $m \nmid y$. Then we measure $\ket{y}$ in either $\ket{\Phi}$ or $\ket{\Phi_h}$ from \cref{shor_eqn2} or \cref{hshor_eqn2} with probability $0$.
\end{lemma} 

\begin{proof} 
	Let $y=m \ell + k$ with $0 < k < \ell$. It suffices to show that $\left| \sum_{c=0}^{m-1}e^{2 \pi i \frac{cd}{2^q}y} \right|^2 = 0.$ Using $\frac{d}{2^q}=\frac 1 m$, we obtain
	\[
	\left|\sum_{c=0}^{m-1} e^{2\pi i \frac{cd}{2^q}y}\right|^2
	=
	\left|\sum_{c=0}^{m-1} e^{2\pi i \frac{c}{m}(m\ell + k)} \right|^2
	=
	\left|\sum_{c=0}^{m-1} \left(e^{2\pi i \frac{k}{m}}\right)^c  \right|^2
	= 0\;.\vspace*{-0.5cm}
	\]
\end{proof}

Hence, we conclude from Lemmata \ref{Hashed-Shor state}, \ref{multiple} and \ref{non-multiple} that in both 
$\ket{\Phi}$ and its hashed version $\ket{\Phi_h}$ we {\em always} measure some $\ket{y}$ for which $y= \ell m = \frac{\ell 2^q}{d}$. 
Assume that $\gcd(\ell, d)=1$, then we directly read off $d$ from $y$. If $\ell$ is uniformly distributed in the interval $[0, d)$ this happens with sufficient probability to compute $d$ in polynomial time. 

Indeed, in \textsc{Shor}'s original algorithm 
$\ell$ is uniformly distributed since the first factor in Eq.~(\ref{shor_eqn2}) satisfies for any $y$

\[
\sum_{k=0}^{d-1}\left|\frac{1}{\sqrt{d}}\cdot e^{2 \pi i \frac{k}{2^q}y}\right|^2 = \frac 1 d  \sum_{k=0}^{d-1}\left| e^{2 \pi i \frac{k}{2^q}y}\right|^2 = \frac 1 d  \sum_{k=0}^{d-1} 1 = 1\;.
\]

Similar to the reasoning in Section~\ref{hashed_simon}, we show that in the case of the hashed version $\ket{\Phi_h}$ we obtain any $\ket{y}$ with $y\not=0$ with probability of at least $\frac{1}{2d}$, where the probability is taken over the random choice of the hash function. This implies that we measure for $\ket{\Phi_h}$ the useless $y=0$ with at most probability $1-\frac{d-1}{2d} \approx \frac 1 2$.

\begin{theorem}\label{measure_probability_2^r}
	Let $N \in \mathbb{N}$, $a \in \Z_N^*$ with $d=\textrm{ord}_N(a)$ a power of two and $f(x) =a^x \bmod N$.
	Let $\mathcal{H}_t=\{h: \Z_N \rightarrow \HD^t\} $ be universal. Then we measure in \textsc{Hashed-Shor} in the $q$ input qubits any $y= \ell m$, $0 < \ell < d$ with probability $\frac{1-2^{-t}}{d}$, where the probability is taken over the random choice of $h \in \mathcal{H}_t$.
\end{theorem}

	\begin{proof}
		Let us denote by $p_h = \Pr_{h \in \mathcal{H}_t}[y]$ the probability that we measure $y$ in \textsc{Hashed-Shor} in the $q$ input qubits. By~\cref{Hashed-Shor state}, Eq.~(\ref{hshor_eqn2}) and \Cref{multiple} we know that for all $y = \ell m = \ell \frac{2^q}{d}$ we have 
		\[
		p_h =\frac{1}{|\mathcal{H}_t|}\sum_{h\in\mathcal{H}_t} \sum_{z\in\HD^t} \left| \frac{1}{\sqrt{d}}\sum_{k\in M_{z}} e^{2\pi i \frac{k}{2^q}y}\right|^2\!\cdot\frac{1}{d} =
		\frac{1}{d^2}\cdot 
		\frac{1}{|\mathcal{H}_t|}\sum_{h\in\mathcal{H}_t}\sum_{z\in\HD^t}\left| \sum_{k\in M_{z}} e^{2\pi i \frac{k \ell}{d}}\right|^2\!\!.
		\]
		Recall that $M_z := \{k \in \Z_d \mid h(a^k \bmod N) = z\}$. Observe that for $k_1 \not= k_2$ we obtain a cross-product $e^{2\pi i \frac{k_1 \ell}{d}}\cdot\overline{e^{2\pi i \frac{k_2 \ell}{d}}} = e^{2\pi i \frac{(k_1-k_2)\ell}{d}}$ iff $k_1, k_2$ are in the same set $M_{z}, z \in \HD^t$, i.e. iff $h(a^{k_1} \bmod N)=h(a^{k_2} \bmod N)$. Using \Cref{universal} of a universal hash function family, we obtain $\Pr_{h \in \mathcal{H}_t}[h(a^{k_1} \bmod N)=h(a^{k_2} \bmod N)]=2^{-t}$ for any $k_1 \not= k_2$. This implies that for exactly $2^{-t}$ of all $h \in \mathcal{H}_t$ we obtain $h(a^{k_1} \bmod N)=h(a^{k_2} \bmod N)$. Therefore,
		\[
		p_h = \frac{1}{d^2}\cdot 
		\left(
		\sum_{k=0}^{d-1} \left|e^{2\pi i \frac{k \ell}{d}}\right|^2 
		+ 2^{-t}
		\sum_{k_1 \in \Z_d} \sum_{k_2\not=k_1\in \Z_d}e^{2\pi i \frac{(k_1-k_2) \ell }{d}}
		\right)\;.
		\]
		Since $k_1 - k_2 \not=0$, we can rewrite as
		\[
		p_h = \frac{1}{d^2}\cdot 
		\left(
		d
		+ \frac{d}{2^t}
		\sum_{k\in \Z_d\setminus\{0\}}e^{2\pi i \frac{k \ell }{d}}
		\right)
		=
		\frac{1}{d^2}\cdot 
		\left(
		d 
		- \frac{d}{2^t}
		\right)
		=
		\frac{1-2^{-t}}{d}\;.\vspace*{-0.5cm}
		\]
	\end{proof}

From \Cref{measure_probability_2^r} we see that in the hashed version $\ket{\Phi_h}$ we measure every $y = m\ell, y \not= 0$ with probability $\frac{1-2^{-t}}{d}$, whereas in comparison in $\ket{\Phi}$ we measure every $y= m\ell$ with probability $\frac 1 d$. It follows that in Eq.~(\ref{hshor_eqn2}) the scaling factor 
\begin{equation}
\label{scaling}
\mathcal{S} = \sum_{z\in\HD^t}\left\vert \frac{1}{\sqrt{d}} \sum_{k\in M_z} e^{2\pi i \frac{k}{2^q}y}\right\vert^2
\end{equation}
takes on expected value $1-2^{-t}$ for $y = \ell m$, $0 < \ell < d$ taken over all $h \in \mathcal{H}_t$. Notice that $\mathcal{S}$ is a symmetric function in $y$, i.e. $\mathcal{S}(y) =  \mathcal{S}(2^q-y)$.

Let us look at an example to illustrate how the probabilities behave.
We choose $N=51=3\cdot 17$, $a=2$ and $q=12$. This implies $d=\textrm{ord}_N(a)=8$ and $m = \frac{2^q}{d} = 512$. In $\Ket{\Phi}$ we measure some $y = m \ell = 512 \ell$, $0 \leq \ell < d = 8$ with probability $\frac 1 8$ each, as illustrated in \Cref{fig:shor}.

Let us assume we have in \textsc{Hashed-Shor} $M_0=\{2,3,4,7\}$ (using $t=1$). This fully specifies the scaling function $\mathcal{S}$ from Eq.~(\ref{scaling}). Thus, each amplitude from $\Ket{\Phi}$ is multiplied by $\mathcal{S}$, as illustrated in \Cref{fig:hashed_shor}.

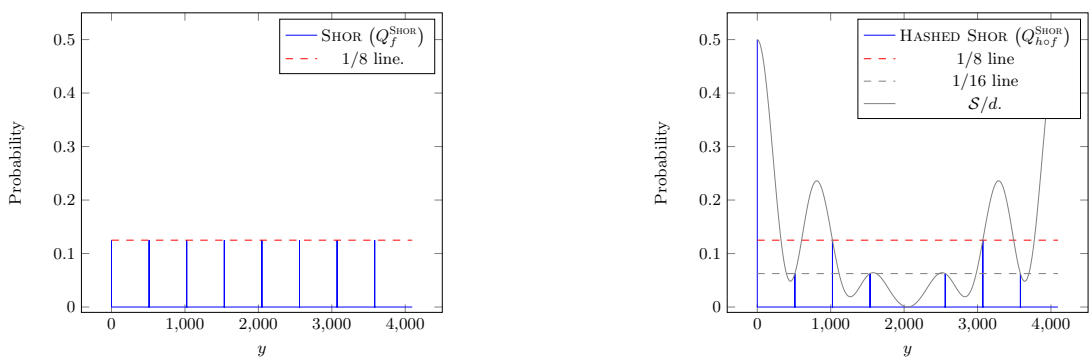
\begin{figure}[H]
	\begin{center}
		\begin{subfigure}[t]{0.48\textwidth}
			\centering
			\begin{tikzpicture}[scale=0.70]
			\begin{axis}[
			legend style={cells={align=left}},
			ymin=-0.01, ymax=0.55,
			xlabel=$y$,
			ylabel=Probability,
			]
			
			\addplot+[mark=none] coordinates {
				(0,0.125)
				(0,0)	
				(512,0)
				(512,0.125)
				(512,0)					
				(1024,0)
				(1024,0.125)
				(1024,0)						
				(1536,0)
				(1536,0.125)
				(1536,0)									
				(2048,0)
				(2048,0.125)
				(2048,0)				
				(2560,0)
				(2560,0.125)
				(2560,0)			
				(3072,0)
				(3072,0.125)
				(3072,0)			
				(3584,0)
				(3584,0.125)
				(3584,0)			
				(4095,0)					
			};		
			\addlegendentry{\textsc{Shor} $\left(Q^\Shor_{f}\right)$}
			\addplot+[red,dash pattern=on 4pt off 4pt, mark=none] coordinates {
				(0,0.125)						
				(4095,0.125)			
			};
			\addlegendentry{$1/8$ line.}
			\end{axis}
			\end{tikzpicture}%
			\caption{\textsc{Shor} with $d=8$, $m=512$. The probabilities are independent of the measurement of the output qubits $z_k$.}\label{fig:shor}
		\end{subfigure}\hfill%
		\begin{subfigure}[t]{0.48\textwidth}
			\centering
			\begin{tikzpicture}[scale=0.70]
			\begin{axis}[
			legend style={cells={align=left}},
			domain=0:4096,
			samples=4096,
			ymin=-0.01, ymax=0.55,
			xlabel=$y$,
			ylabel=Probability,
			]
			
			\addplot+[mark=none] coordinates {
				
				(0,0.5)
				(0,0)	
				(512,0)
				(512,0.0625)
				(512,0)					
				(1024,0)
				(1024,0.125)
				(1024,0)						
				(1536,0)
				(1536,0.0625)
				(1536,0)									
				(2048,0)
				(2048,0.0)
				(2048,0)				
				(2560,0)
				(2560,0.0625)
				(2560,0)			
				(3072,0)
				(3072,0.125)
				(3072,0)			
				(3584,0)
				(3584,0.0625)
				(3584,0)			
				(4095,0)			
			};
			\addlegendentry[align=center]{\textsc{Hashed Shor} $\left(Q^\Shor_{h\circ f}\right)$}
			\addplot+[red,dash pattern=on 4pt off 4pt, mark=none] coordinates {
				(0,0.125)					
				(4096,0.125)			
			};
			\addlegendentry{$1/8$ line}
			\addplot+[gray,dash pattern=on 4pt off 4pt, mark=none] coordinates {
				(0,0.0625)					
				(4096,0.0625)			
			};
			\addlegendentry{$1/16$ line}
			\addplot+[gray,mark=none, 
			samples=256] expression{
				(cos(deg(2*pi*0/4096*x))
				+cos(deg(2*pi*1/4096*x))
				+cos(deg(2*pi*2/4096*x))
				+cos(deg(2*pi*5/4096*x))
				+cos(deg(2*pi*0/4096*x))
				+cos(deg(2*pi*1/4096*x))
				+cos(deg(2*pi*1/4096*x))
				+cos(deg(2*pi*4/4096*x))
				+cos(deg(2*pi*2/4096*x))
				+cos(deg(2*pi*1/4096*x))
				+cos(deg(2*pi*0/4096*x))
				+cos(deg(2*pi*3/4096*x))
				+cos(deg(2*pi*5/4096*x))
				+cos(deg(2*pi*4/4096*x))
				+cos(deg(2*pi*3/4096*x))
				+cos(deg(2*pi*0/4096*x)))/(8*8)
				+(cos(deg(2*pi*0/4096*x))
				+cos(deg(2*pi*1/4096*x))
				+cos(deg(2*pi*5/4096*x))
				+cos(deg(2*pi*6/4096*x))
				+cos(deg(2*pi*1/4096*x))
				+cos(deg(2*pi*0/4096*x))
				+cos(deg(2*pi*4/4096*x))
				+cos(deg(2*pi*5/4096*x))
				+cos(deg(2*pi*5/4096*x))
				+cos(deg(2*pi*4/4096*x))
				+cos(deg(2*pi*0/4096*x))
				+cos(deg(2*pi*1/4096*x))
				+cos(deg(2*pi*6/4096*x))
				+cos(deg(2*pi*5/4096*x))
				+cos(deg(2*pi*1/4096*x))
				+cos(deg(2*pi*0/4096*x)))/(8*8)
			};		
			\addlegendentry{$\mathcal{S}/d$.}
			\end{axis}
			\end{tikzpicture}
			\caption{\textsc{Hashed-Shor} with $d=8$, $m=512$, $t=1$ and $M_0=\{2,3,4,7\}$.}\label{fig:hashed_shor}
			
		\end{subfigure}%
	\end{center}
	\caption{Probability distributions for \textsc{Shor} and \textsc{Hashed-Shor}}
	\label{fig:shor_diagram}
\end{figure}

\begin{theorem}\label{hshor_2^r}
	Let $N \in \mathbb{N}$, $a \in \Z_N^*$ with $d=\textrm{ord}_N(a)$ a power of two and $f(x)=a^x \mod N$.
	Let $\mathcal{H}_t=\{h: \Z_N \rightarrow \HD^t\} $ be universal, and let $x$ be represented by $q$ input qubits in $Q^\Shor_{h\circ f}$. Then \textsc{Hashed-Shor} finds $f$'s period $d$ with expected $\frac{2}{1-2^{-t}}\le4$ applications of quantum circuits $Q^\Shor_{h\circ f}$, $h \in_R\mathcal{H}_t$, that use only $q+t$ qubits.
\end{theorem}

	\begin{proof}
		In \textsc{Shor} we compute the fraction $\frac y {2^q} = \frac{\ell}{d}$ in reduced form. Since $d$ is a power of two, this fraction reveals $d$ in its denominator iff $\ell$ is odd. Using \cref{measure_probability_2^r}, we measure $y = m \ell$ with an odd $\ell$, $0 < \ell < d$ with probability $\frac{d}{2} \cdot \frac{1-2^{-t}}{d} = \frac{1-2^{-t}}{2}$. Thus, we need on expectation $\frac{2}{1-2^{-t}}\le 4$ applications of $Q^\Shor_{h\circ f}$ to find $f$'s period $d$.
		
		Notice that we can check the validity of $d$ via testing the identity $a^d \stackrel{?}{=} 1 \bmod N$.
	\end{proof}

For comparison, we need in \textsc{Shor}'s original algorithm with the non-hashed version of $f$ on expectation $2$ measurements until we find $d$.

%
%

\section{Hashed Period-Finding Including Shor}
\label{shor_general}
Notice that we proved in \cref{measuring_probabilities} and \cref{measure_probability_2^r} that when we move to the hashed version of our quantum circuits all probabilities to measure some $y \not=0$ decrease exactly by a factor of $(1-2^{-t})$  (over the random choice of the $t$-bit hash function). 

The same is true for finding arbitrary (non power of two) periods  with circuit $Q^{\Shor}_{h\circ f}$. However, this does not immediately follow from the proof of \cref{measure_probability_2^r}, because the proof builds on the special form of superposition $\Ket{\Phi_h}$ from Eq.~(\ref{hshor_eqn2}) that only holds if $d$ is a power of two. 
Here we show a more general result for period finding algorithms that applies for Shor's original circuit as well as for its Eker\aa -H\aa stad variant in the subsequent section. To this end let us define a generic period finding quantum circuit $Q_f^{\Period}$ (see \cref{circuit:period}). In \cref{circuit:period} we denote by $Q_1, Q_2$ any quantum circuitry that acts on the $q$ input qubits. For example, for Simon's circuit we have $Q_1 = Q_2 = H_q$ (see \cref{circuit:simon}). For Shor's circuit we have $Q_1 = H_q$ and $Q_2 = \text{QFT}_q$. In the following \cref{main_theorem} we define explicitly a {\em cancellation criterion} that this circuitry $Q_1, Q_2$ has to fulfill. An important feature of $Q_f^{\Period}$ is however that we apply $f$ only once.

\begin{figure}[h]
	\centering
	\mbox{
		\Qcircuit @C=1em @R=1em {
			\lstick{\ket{0^{q}}}	& \qw & \gate{\text{Q}_1} 	& \qw & \multigate{1}{U_{f}}& \qw & \gate{\text{Q}_{2}} 	& \qw & \meter\\
			\lstick{\ket{0^{n}}}	& \qw & \qw	 	   			& \qw & \ghost{U_{f}}      	& \qw & \qw	    		 		& \qw & \qw
		}
	}
	\caption{Quantum circuit $Q^{\Period}_{f}$}\label{circuit:period}
\end{figure}
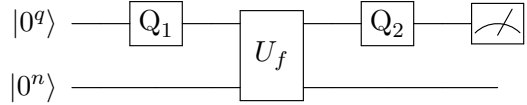
Now let us use our generic period finding circuit $Q_f^{\Period}$ inside a generic period finding algorithm \textsc{Period}  that uses a certain number of measurements of $Q_f^{\Period}$ and some classical post-processing. If we replace in \textsc{Period} the circuit $Q_f^{\Period}$ by its hashed variant $Q_{h \circ f}^{\Period}$ then we call the resulting algorithm \textsc{Hashed-Period} (Algorithm~\ref{alg:period}).
\begin{algorithm}[ht]
	\DontPrintSemicolon
	\SetAlgoLined
	\SetKwInOut{Input}{Input}\SetKwInOut{Output}{Output}
	
	\SetKwComment{COMMENT}{$\triangleright$\ }{}
	
	\Input{ $f: \HD^q  \rightarrow \HD^n$, universal $\mathcal{H}_t:= \{h: \HD^n \rightarrow \HD^t\}$}
	
	\Output{ Period $d$ of $f$}
	
	\Begin{
		Set $Y=\emptyset$.\; 
		\Repeat{$|Y|$ is sufficiently large.}{
			Run $Q^{\Period}_{h \circ f}$ on $\ket{0^q}\ket{0^t}$ for some freshly chosen $h \in_R \mathcal{H}_t$. \;
			Let $y$ be the measurement of the $q$ input qubits. \;
			If $y \not= 0^q$, then include $y$ in $Y$.
		}
		Compute $d$ from $Y$ in a classical post-process.\;
		\KwRet{$d$}
	}
	\caption{\textsc{Hashed-Period}}\label{alg:period}
\end{algorithm}

Notice that \textsc{Period} can be considered as special case of \textsc{Hashed-Period}, where we choose $t=n$ and the identity function $h=id$. This slightly abuses notation, since $\mathcal{H}_n=\{id\}$ is not universal.

The proof of the following theorem closely follows the reasoning in the proof of \cref{measuring_probabilities}. Here we show that the probabilities decreases by exactly a factor of $1-2^{-t}$ in the hashed version if a certain {\em cancellation criterion} (\cref{eq:cancel}) is met. 

\begin{theorem}\label{main_theorem}
	Let $f: \HD^q \rightarrow \HD^n$ and $\mathcal{H}_t=\{h: \HD^n \rightarrow \HD^t\} $ be universal. Let $Q^{\Period}_f$ be a quantum circuit that on input $\ket{0^q}\ket{0^n}$ yields a superposition
	\begin{equation}
	\label{eq:cancel}
		\ket{\Phi} = \sum_{y \in \HD^q} \sum_{f(x) \in \textrm{Im}(f)} w_{y,f(x)} \ket{y}\ket{f(x)} \textrm{ satisfying } \sum_{f(x) \in \textrm{Im}(f)} w_{y,f(x)} = 0 \textrm{ for any } y\not= 0. 
	\end{equation}
	Let us denote by $p(y)$, respectively $p_h(y)$, the probability to measure some $\ket{y}$, $y \not= 0$ in the $q$ input qubits when applying $Q^\Period_f$, respectively $Q^\Period_{h \circ f}$ with $h \in_R \mathcal{H}_t$. 
	Then
	$p_h(y) = (1-2^{-t})\cdot  p(y).$
\end{theorem} 

	\begin{proof}
		For ease of notation let us denote $z=f(x)$. By definition, we have  $p(y) = \sum_{z \in \textrm{Im}(f)} |w_{y,z}|^2$.
		
		Now let us find an expression for $p_h(y)$ when using $Q_{h \circ f}^\Period$.
		For $h \in \mathcal{H}_t$ we denote $I_{h,z'}=\{ z \in \textrm{Im}(f) \mid h(z)=z'\}$.
		Since $\mathbin{\dot{\bigcup}}_{z'\in \F_2^t} I_{h, z'} = \textrm{Im}(f)$, we obtain 
		\begin{align}
			p_h(y)
			& = \frac{1}{|\mathcal{H}_t|}\sum_{h\in\mathcal{H}_t}\sum_{z'\in \textrm{Im}(h)}\left|\sum_{z \in I_{h,z'}} w_{y,z}\right|^2. \label{cross3}
		\end{align}
		In Eq.~(\ref{cross3}) we obtain a cross-product $w_{y,z_1}\overline{w_{y,z_2}}$ for $z_1 \not=z_2$ iff $z_1, z_2$ are in the same set $I_{h,z'}$, $z'\in \textrm{Im}(h)$, i.e. iff $h(z_1)=h(z_2)$. Using \Cref{universal} of a universal hash function family, we obtain $\Pr_{h \in \mathcal{H}_t}[h(z_1)=h(z_2)]=2^{-t}$ for any $z_1 \not= z_2$. This implies that for exactly $2^{-t}$ of all $h \in \mathcal{H}_t$ we obtain $h(z_1)=h(z_2)$. 
		We conclude that
		\[
		p_h(y) =\sum_{z \in \textrm{Im}(f)} |w_{y,z}|^2+2^{-t}\cdot\sum_{z_1\not=z_2} w_{y,z_1} \overline{w_{y,z_2}}.
		\]
		Our prerequisite $\sum_{z \in \textrm{Im}(f)} w_{y,z} = 0$ for any $y \not= 0$ implies
		\begin{align*}
			0 &=2^{-t}\left|\sum_{z \in \textrm{Im}(f)} w_{y,z}\right|^2 =2^{-t} \sum_{z \in \textrm{Im}(f)} \left|w_{y,z}\right|^2+2^{-t}\sum_{z_1\not=z_2} w_{y,z_1}\overline{w_{y,z_2}}\\ &= p_h(y) - (1-2^{-t})\cdot \sum_{z \in \textrm{Im}(f)} |w_{y,z}|^2.	
		\end{align*}  
		Together with the definition of $p(y)$ we conclude that 
		\[
		p_h(y) 
		=  (1-2^{-t})\cdot \sum_{z \in \textrm{Im}(f)} |w_{y,z}|^2
		=  (1-2^{-t})\cdot  p(y).
		\]
	\end{proof}

We already showed in \cref{Simon_Sum_Zero} that Simon's circuit $Q^{\Simon}_f$ fulfills the cancellation criterion (\cref{eq:cancel}) of \cref{main_theorem}. Thus, the statement of \cref{measuring_probabilities} directly follows from \cref{main_theorem}. However, for an improved intelligibility we preferred to prove \cref{measuring_probabilities} directly.

In the following \cref{main_theorem_shor} we show that $Q^{\Shor}_f$ also meets the cancellation criterion. Thus, going to the hashed version in Shor's algorithm immediately scales all probabilities by a factor of $1-2^{-t}$ for $y \not= 0$.

\begin{lemma}\label{main_theorem_shor}
	On input $\ket{0^q}\ket{0^n}$ the quantum circuit $Q^{\Shor}_{f}$  yields a
	superposition
	\[
		\ket{\Phi} = \sum_{y \in \HD^q} \sum_{f(x) \in \textrm{Im}(f)} w_{y,f(x)} \ket{y}\ket{f(x)} \textrm{ satisfying } \sum_{f(x) \in \textrm{Im}(f)} w_{y,f(x)} = 0 \textrm{ for any } y\not= 0. 
	\]
\end{lemma}

\begin{proof}
	From \cref{shor_eqn} we know that $Q_{f}^{\Shor}$ yields a superposition
	\[
	\ket{\Phi}=
	\sum_{y=0}^{2^q-1} \sum_{f(x) \in \textrm{Im}(f)} w_{y,f(x)} \ket{y}\ket{f(x)} \textrm{ with } \sum_{f(x) \in \textrm{Im}(f)} w_{y,f(x)} = \sum_{k=0}^{d-1} \frac{1}{2^q} \sum_{c \geq 0: \atop cd+k < 2^q}e^{2 \pi i \frac{cd+k}{2^q}y}\;.
	\]
	We conclude for $y\neq 0$ that
	\[
	\sum_{f(x) \in \textrm{Im}(f)} w_{y,f(x)} = \sum_{k=0}^{d-1} \frac{1}{2^q} \sum_{c \geq 0: \atop cd+k < 2^q}e^{2 \pi i \frac{cd+k}{2^q}y} 
	=\sum_{r=0}^{2^q-1}\frac{1}{2^q}e^{2 \pi i \frac{r}{2^q}y}
	=\frac{1}{2^q}\sum_{r=0}^{2^q-1}\left(e^{2 \pi i \frac{y}{2^q}}\right)^r
	=0\;.	
	\] 	
\end{proof}

Since by \cref{main_theorem} the use of hashed versions at most halves all probabilities for $y \not= 0$, we expect that \textsc{Hashed-Period} requires at most twice as many measurements as \textsc{Period}. This is more formally shown in the following \cref{main_theorem_period}.

\begin{theorem}\label{main_theorem_period}
	Let $f: \HD^q \rightarrow \HD^n$ have period $d$, and let $\mathcal{H}_t=\{h: \HD^n \rightarrow \HD^t\} $ be universal.
	Assume that \textsc{Period} succeeds to find $d$ with probability $\rho$ with an expected number of $m$ measurements, using some $Q^{\Period}_{f}$ with $q+n$ qubits that satisfies the cancellation criterion (\cref{eq:cancel}) of \cref{main_theorem}. Then \textsc{Hashed-Period} succeeds to find $d$ with probability $\rho$ using $Q^{\Period}_{h \circ f}$, $h \in_R \mathcal{H}_t$, with only $q+t$ qubits and an expected number of $\frac{m}{1-2^{-t}}$ measurements.	
\end{theorem}

	\begin{proof}
		We first show the factor of $\frac{1}{1-2^{-t}}$ difference in the expected number of measurements.
		In the case of $Q^{\Period}_{f}$ we measure some $y \not = 0$ with probability $\sum_{y =1}^{2^q-1} p(y)$, whereas for $Q^{\Period}_{h \circ f}$ we measure $y \not=0$ with $1-2^{-t}$ times the probability $\sum_{y =1}^{2^q-1} (1-2^{-t})\cdot p(y) = (1-2^{-t}) \cdot\sum_{y =1}^{2^q-1} p(y)$ according to \cref{main_theorem}. This implies that on expectation we need $\frac{1}{1-2^{-t}}$ as many measurements.
		
		It remains to show that \textsc{Hashed-Period} has the same success probability $\rho$ as \textsc{Period} to compute the period $d$. To this end we show that conditioned on $y \not= 0$, both circuits $Q^{\Period}_{f}$ and $Q^{\Period}_{h \circ f}$ yield an identical probability distribution for the measured $\ket{y}$ in the $q$ input qubits.
		
		Let $p(y)$, respectively $p_h(y)$, be the probability that we measure $\ket{y}$ in the $q$ input qubits using $Q^{\Shor}_{f}$, respectively $Q^{\Shor}_{h \circ f}$. 
		Since \textsc{Period} conditions on measuring $y \not=0$ we obtain in the case of $Q^{\Shor}_{f}$ the probabilities
		\[
		\frac{p(y)}{\sum_{y =1}^{2^q-1} p(y)} \textrm{ for any } y\not=0.
		\]
		In the case $Q^{\Shor}_{h \circ f}$, we obtain using \cref{main_theorem} the same probabilities
		\[
		\frac{(1-2^{-t})\cdot p(y)}{\sum_{y =1}^{2^q-1} (1-2^{-t})\cdot p(y)} = \frac{p(y)}{\sum_{y =1}^{2^q-1} p(y)} \textrm{ for any } y\not=0.
		\]
		Since both probability distributions are identical, the success probability $\rho$ is identical as well, independent of any specific post-process for computing $d$.
	\end{proof}

Since by \cref{main_theorem_shor} Shor's circuit $Q_f^{\Shor}$ satisfies the cancellation criterion of \cref{main_theorem}, \cref{main_theorem_period} implies that we can implement Shor's algorithm oracle-based with $q+t$ instead of $q+n$ qubits at the cost of only $\frac{1}{1-2^{-t}}$ times as many measurements. In other words, for $t=1$ we save all but one of the output qubits at the cost of twice as many measurements.

%
%

%
%
\section{Oracle-Based Hashed Eker\aa -H\aa stad}
\label{ekera_sec}

In 2017, Eker\aa\ and H\aa stad~\cite{DBLP:journals/corr/EkeraH17} proposed a variant of Shor's  algorithm for computing the discrete logarithms of $x=g^d$ in polynomial time with only $(1+o(1))\log d$ input qubits. The Eker\aa -H\aa stad algorithm saves input qubits in comparison to Shor's original discrete logarithm algorithm whenever $d$ is significantly smaller than the group order. 

An interesting application of such a small discrete logarithm algorithm is the factorization of $n$-bit RSA moduli $N=pq$, where $p,q$ are primes of the same bit-size. Let $g \in_R \Z_N^*$. Then $\textrm{ord}_N(g)$ divides $\phi(N)/2 = (p-1)(q-1)/2 = \frac{N+1}{2} - \frac{p+q}{2}$. Therefore
\[
  	x:= g^{\frac{N+1}{2}} = g^{\frac{p+q}{2}} \bmod N.
\]  
Hence, we obtain a discrete logarithm instance in $\Z_N^*$ where the desired logarithm $d = \frac{p+q}{2}$ is of size only roughly $\frac n 2$ bits, whereas group elements have to be represented with $n$ bits. Notice that the knowledge of $d=\frac{p+q}{2}$ together with $N=pq$ immediately yields the factorization of $N$ in polynomial time.

The Eker\aa-H\aa stad algorithm computes $d$ with $(\frac 1 2 + \frac 1 s)n$ input and $n$ output qubits, using a classical post-process that takes time polynomial in $n$ and $s^s$. Choosing $s = \frac{\log n}{\log \log n}$, we obtain a polynomial time factoring algorithm with a total of $(\frac 3 2 + o(1))n$ qubits. 

\medskip

In the following, we show that the Eker\aa-H\aa stad algorithm is covered by our framework of quantum period finding algorithms which fulfill the cancellation criterion of \cref{eq:cancel} from \cref{main_theorem}. Thus, by \cref{main_theorem_period} we can save all but $1$ of the $n$ output qubits via (oracle-based) hashing, at the cost of only doubling the number of quantum measurements. This in turn leads to a polynomial time (oracle-based) factorization algorithm for $n$-bit RSA numbers using only $(\frac 1 2 + o(1))n$ qubits. 
Concerning discrete logarithms, with our (oracle-based) hashing approach we can quantumly compute $d$ from $g$ and $g^d$ in polynomial time using only $(1 + o(1))\log d$ qubits.

Let $(g, x=g^d, S(G))$ be a discrete logarithm instance with $m=\log d$.  Here $S(G)$ specifies how we compute in the group $G$ generated by $g$,  e.g. $S(G)=N$ specifies that we compute modulo $N$ in the group $G=\Z_N^*$.
Define 
\[
  	f_{g,x,S(G)}(a,b) = a^x \cdot x^{-b} = g^{a-bd}.  
\]
The Eker\aa-H\aa stad quantum circuit $Q_{f}^{\EHS}$ from \cref{circuit:ekeraa} computes on input $\ket{0^{m +\ell}}\ket{0^{\ell}}\ket{0^n}$, where $\ell:=\frac m s$,  a superposition 
\begin{equation}
	\label{Ekera_state}
	\ket{\Phi}=\frac{1}{{2^{m + 2 \ell}}}\sum_{a,j=0}^{2^{m + \ell}-1}\sum_{b,k=0}^{2^{\ell}-1}
	e^{2 \pi i (aj + 2^m b k ) / 2^{m +\ell}} \ket{j,k,f_{g,x,S(G)}(a,b)} \ . 
\end{equation}

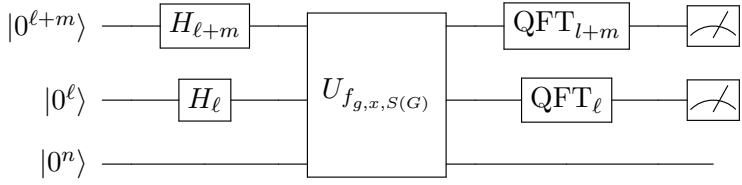
\begin{figure}[h]
	\centering
	\mbox{
		\Qcircuit @C=1em @R=1em {
			\lstick{\ket{0^{\ell+m}}}	& \qw & \gate{H_{\ell+m}} 	& \qw & \multigate{2}{U_{f_{g,x,S(G)}}} & \qw & \gate{\text{QFT}_{l+m}} & \qw & \meter\\
			\lstick{\ket{0^{\ell}}}	& \qw & \gate{H_{\ell}} 	& \qw & \ghost{U_{f_{g,x,S(G)}}}  		& \qw & \gate{\text{QFT}_{\ell}} 	& \qw & \meter\\
			\lstick{\ket{0^{n}}}	& \qw & \qw	 	   		& \qw & \ghost{U_{f_{g,x,S(G)}}}        & \qw & \qw	    		 		& \qw & \qw
		}
	}
	\caption{Quantum circuit $Q^{\EHS}_{f}$}\label{circuit:ekeraa}
\end{figure}

The main step in the analysis of  Eker\aa\--H\aa stad shows that we measure in the $m +2\ell = (1+ \frac 2 s)m= (1+ \frac 2 s) \log d$ input qubits with high probability so-called {\em good pairs} $(j,k)$  that help us in computing $d$ via some lattice reduction technique. 

In the following \cref{main_theorem_ekera}, we show that $Q_{f}^{\EHS}$ satisfies our cancellation criterion of \cref{main_theorem}. Thus, we conclude from \cref{main_theorem} that by moving to the $1$-bit hashed version $Q_{h \circ f}^{\EHS}$ we lower the probabilities of measuring good $(j,k)$ only by a factor of $\frac 1 2$ (averaged over all hash functions).

\begin{lemma}\label{main_theorem_ekera}	
	Let $(g,x,S(G))$ be a discrete logarithm instance and  \linebreak$f_{g,x,S(G)}(a,b)=g^{a} x^{-b}$.
	On input $\ket{0^{2\ell+m}}\ket{0^n}$ the quantum circuit $Q^{\EHS}_{f}$  yields a
	superposition
	\[
		\ket{\Phi} = \sum_{y \in \HD^{m + 2\ell}} \sum_{f(x) \in \textrm{Im}(f)} w_{y,f(x)} \ket{y}\ket{f(x)} \textrm{ satisfying } \sum_{f(x) \in \textrm{Im}(f)} w_{y,f(x)} = 0 \textrm{ for any } y\not= 0\;. 
	\]
\end{lemma}

\begin{proof}
	From Eq.~(\ref{Ekera_state}) with $y=(j,k)$ we know that $Q_{f}^{\EHS}$ yields a superposition
	\begin{align*}
	& \ket{\Phi}=
	\sum_{j=0}^{2^{m +\ell}-1} \sum_{k=0}^{2^{m + \ell}-1}  \sum_{f(x) \in \textrm{Im}(f)} w_{(j,k),f(x)} \ket{j,k}\ket{f(x)} \textrm{ with } \\ 
	& \sum_{f(x) \in \textrm{Im}(f)} w_{(j,k),f(x)} = \frac{1}{{2^{m + 2\ell}}}\sum_{a=0}^{2^{m + \ell}-1}\sum_{b=0}^{2^{\ell}-1}
	e^{2 \pi i (aj + 2^m b k ) / 2^{m + \ell}}.
	\end{align*}
	Hence for $y\neq 0$ we obtain
	\begin{align*}
	\sum_{f(x) \in \textrm{Im}(f)} w_{(j,k),f(x)} 	& =
	\frac{1}{2^{m + 2\ell}}
	\left(\sum_{a=0}^{2^{m + \ell}-1}e^{2\pi i aj/2^{m + \ell}}\right)\cdot\left(\sum_{b=0}^{2^{\ell}-1}e^{2\pi i bk/2^{\ell}}\right)\\
	& =
	\frac{1}{2^{m + 2\ell}}
	\left(\sum_{a=0}^{2^{m + \ell}-1}\left(e^{2\pi i j/2^{m + \ell}}\right)^a\right)\cdot\left(\sum_{b=0}^{2^{\ell}-1}\left(e^{2\pi i k/2^{l}}\right)^b\right) \;.
	\end{align*}
	Since by prerequisite $(j,k) \not=(0,0) \in \Z_{2^{m + \ell}} \times \Z_{2^{\ell}}$, we have $j \not= 0 \bmod 2^{m + \ell}$ or $k \not= 0 \bmod 2^{\ell}$. This implies that at least one of the factors is identical $0$.
\end{proof}

By \cref{main_theorem_period},  replacing in the Eker\aa -H\aa stad algorithm the quantum circuit $Q_{f}^{\EHS}$ by single output bit circuits $Q_{h \circ f}^{\EHS}$  comes at the cost of only twice the number of measurements. Since the Eker\aa -H\aa stad algorithm finds discrete logarithms $d$ in polynomial time using only\linebreak $m + 2\ell = (1 + o(1))\log d$ input qubits, we obtain from \cref{main_theorem_period} the following corollary.

\begin{corollary}
	Eker\aa -H\aa stad's Shor variant admits an {\em oracle-based} hashed version that 
	\begin{enumerate}
		\item computes discrete logarithms $d$ from $g, g^d$ in polynomial time using\linebreak $(1+o(1))\log d$ qubits,  
		\item factors $n$-bit RSA numbers in time polynomial in $n$ using $(\frac 1 2 +o(1)) n$ qubits.
	\end{enumerate}
\end{corollary}
\noindent {\bf Open Problem:} Can we modify our oracle-based approach into a real-world application similar to the results in \cref{sec:even_mansou} for the Simon algorithm? That is, can we define (not necessary single bit) hashed versions of the exponentiation function without first computing the full function value? 

%
%

%
%

%
%


\section{Hashed Mosca-Ekert}\label{sec:mosca_ekert}

Let us briefly recall the Mosca-Ekert variant~\cite{DBLP:conf/qcqc/MoscaE98} of Shor's algorithm that works with a single input qubit. In  \cref{circuit:ME}, we elaborate on Shor's quantum circuit $Q_{f}^{\Shor}$ from \cref{circuit:shor}, where we instantiate the exponentiation function $f:x \mapsto a^x \bmod N$ and $QFT_3$ on three input qubits. The exponentiation is performed bitwise via controlled multiplications with powers of $a^1, a^2, a^4$ via quantum mappings
\[
  \bar U_{a^i} : \{0,1\} \times \Z_N^* \rightarrow \{0,1\} \times \Z_N, (x,y) \mapsto (x, y \cdot (a^i)^x \bmod N).
\]
If we initialize the $n$ output qubits with $\ket{0^{n-1}1}$ -- representing the $1$-element in the multiplicative group $\Z_N^*$ -- circuit $Q_{f}^{\Shor}$ computes $a^x \bmod N$ on its output qubits.

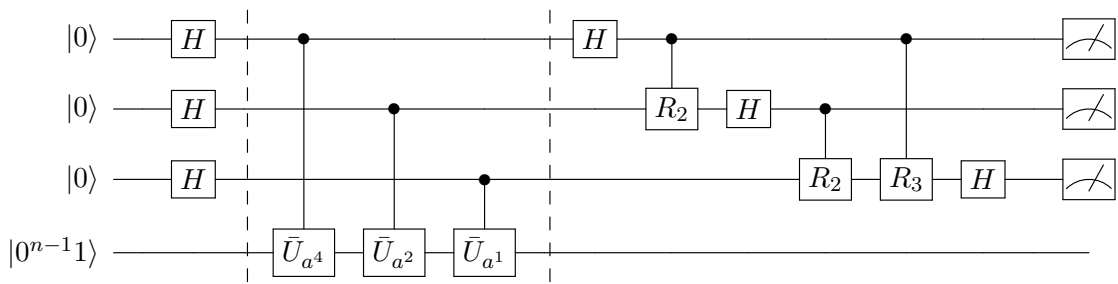
\begin{figure}[htb]
	\centering
	\mbox{
		\Qcircuit @C=1em @R=1em {
			&\lstick{\ket{0}}		& \qw 				 	& \gate{H} & \qw\barrier[-0.75cm]{3} 	&	\ctrl{3}		& \qw 				& 	\qw				& \qw\barrier[-0.6cm]{3} 	& \gate{H} 	& \ctrl{1}		& \qw		& \qw		& \ctrl{2} 	& \qw		& \qw & \meter\\
			&\lstick{\ket{0}}		& \qw 				 	& \gate{H} & \qw						& 	\qw 			& \ctrl{2}			& 	\qw				& \qw					 	& \qw	 		& \gate{R_2}	& \gate{H}& \ctrl{1}  & \qw		& \qw		& \qw & \meter\\	
			&\lstick{\ket{0}}		& \qw 				 	& \gate{H} & \qw						& 	\qw 			& \qw				&	\ctrl{1}		& \qw					 	& \qw	 		& \qw	 		& \qw	 	& \gate{R_2}& \gate{R_3}& \gate{H}& \qw & \meter\\
			&\lstick{\ket{0^{n-1}1}}		& \qw 				 	& \qw		 & \qw						&	\gate{\bar U_{a^4}}	& \gate{\bar U_{a^2}}	&	\gate{\bar U_{a^1}}	& \qw					 	& \qw		 	& \qw 			& \qw		& \qw		& \qw		& \qw		& \qw & \qw
		}
	}
	\caption{Circuit $Q_{f}^{\Shor}$. The $R_j$-gates realize phase shifts $e^{2\pi i/2^j}$.}\label{circuit:ME}
\end{figure}

Mosca and Ekert showed that the computations on the input bits of \linebreak$x=x_1x_2x_3$ can be sequentialized and thus pipelined bit-wise as depicted in \cref{circuit:oneME}. The idea is to measure the first qubit and feed in the result in the controlled rotation $R_2$ for the second bit $x_2$, whose computation is performed again on the same input qubit. Analogous the measurements of the first and second qubits are fed into the controlled rotations $R_2$ and $R_3$ of the third bit $x_3$.

\begin{figure}[htb]
	\centering	
	\includegraphics[width=\textwidth]{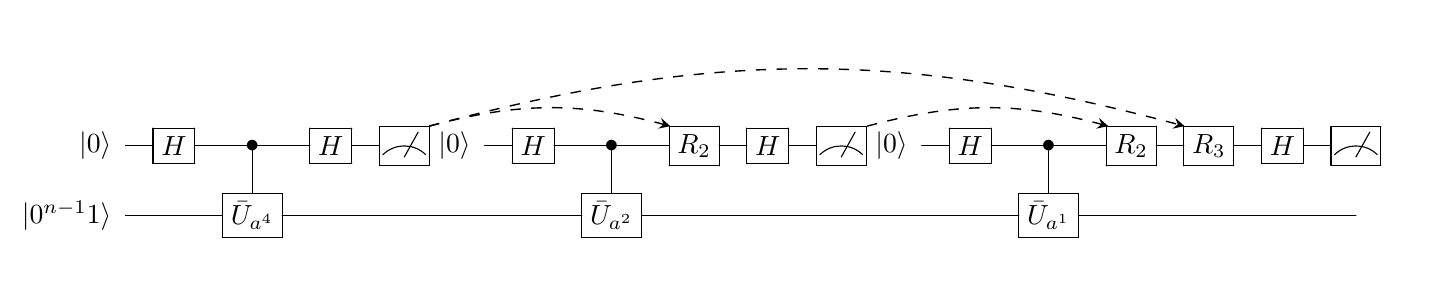}
	\caption{Mosca-Ekert circuit $Q_f^{\textsc{ME}}$}\label{circuit:oneME}
\end{figure}

It is worth noticing that the pipelined processing of the input bits of \linebreak$x=x_1x_2x_3$ leads in $Q_f^{\textsc{ME}}$ also to a pipelined processing of $f$, bit-wise for the $x_i$. Thus, $Q_f^{\textsc{ME}}$ does {\em not} fall into the quantum circuit class $Q_f^{\Period}$ from \cref{circuit:period}, for which we stressed the property that $f$ is only applied once. 

If we still want to directly apply \cref{main_theorem}, we can e.g. assume the existence of some {\em homomorphic} universal hash function family ${\cal H}_t: \Z_N^* \rightarrow (G, \circ)$, where $(G, \circ)$ is a group whose elements are represented by $t$ bits, and for all $h \in  {\cal H}_t$ we have 
\begin{equation}\label{eq:homohash}
  h(a^{x_1}) \circ h(a^{2x_2}) \circ h(a^{4x_3}) = h(a^{x_1} \cdot a^{2x_2} \cdot a^{4x_3} \bmod N) = h(a^x \bmod N) = h(f(x)).
\end{equation}

If this homomorphic property holds, then the circuit $Q_{h \circ f}^{\textsc{ME}}$ depicted in \cref{circuit:HME} can be interpreted as a hashed version of a {\em single application} of $f$, and thus the results of \cref{main_theorem} apply. Notice that in \cref{circuit:HME} we realize the quantum mappings
\[
  \hat U_{h(a^i)} : \{0,1\} \times (G, \circ) \rightarrow \{0,1\} \times (G, \circ), \ (x,y) \mapsto (x, y \circ h((a^i)^x \bmod N)).
\]
Moreover, the output qubits are initialized with $\ket{0^{t}}$, the neutral element of $(G, \circ)$.

\begin{figure}[htb]\vspace*{-1cm}
	\centering
	\includegraphics[width=\textwidth]{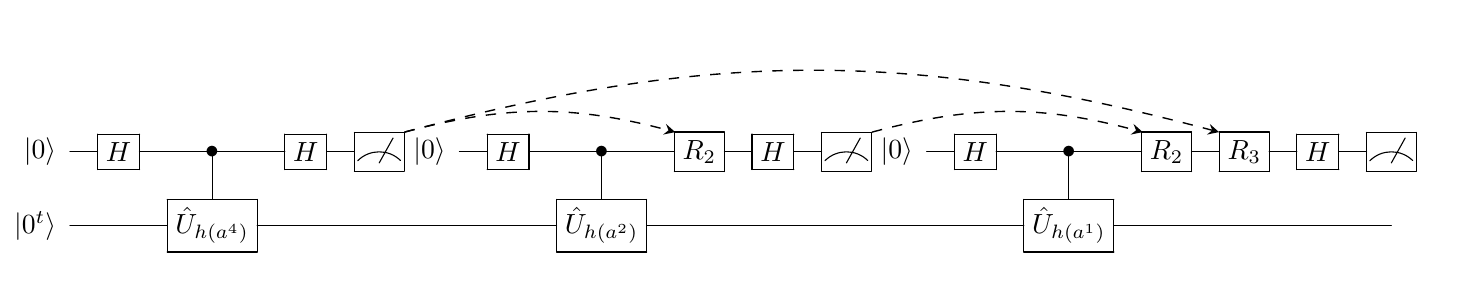}
	\caption{Hashed Mosca-Ekert circuit $Q_{h \circ f}^{\textsc{ME}}$}\label{circuit:HME}
\end{figure}

Therefore, we obtain a polynomial time quantum factorization (or discrete logarithm) algorithm in a {\em non-oracle} setting with only $1+t$ qubits, assuming the existence of an efficient realization of a $t$-bit range homomorphic universal hash function family. 

Let $p$ divide $N$ for some $t$-bit $p$. Then a homomorphic hash function is defined via the canonical ring homomorphism
\[
  h: \Z_N^* \rightarrow \Z_p^*, \ x \bmod N \mapsto x \bmod p.
\]
However, this hash function is useless, because it already assumes that we know a non-trivial factor of $N$.

\vskip 0.2cm

\noindent {\bf Open Problem:} Can we efficiently construct a homomorphic universal hash function family with $t<n$?
 
\vskip 0.2cm 
 
Especially interesting is the case $t= \bigO(\log n)$. In this case our circuit consumes only $\bigO(\log n)$ many qubits, and thus can be simulated in polynomial time on a classical computer.

\begin{corollary}
\label{cor:P}
The existence of an efficiently constructable homomorphic universal hash function family with $\bigO(\log\log N)$-bit range for any $N \in \mathbb{N}$ implies that factoring is in ${\cal P}$. 
An analogous result holds for the discrete logarithm problem in any abelian group.
\end{corollary}

\cref{cor:P} implies that under the assumption that factoring $N$ is hard, respectively that discrete logarithms in some group $G$ are hard, there is no efficiently constructable homomorphic universal hash function family in $\Z_N^*$, respectively $G$, with range logarithmic in the bit-size of the group.


\section{Realization of Hashed Mosca-Ekert -- Solving DDH in $\F_{p^m}$ with a $\frac 1 m$-fraction of qubits }\label{sec:ddh}

In this section we describe a practical realization of our hash technique for the Decisional Diffie Hellman (DDH) problem in $\F_{p^m}$. 

\begin{definition}[DDH problem in $\F_{p^m}$]
	Let $\F_{p^m}^*$ be the multiplicative group of a finite field such that $q_1=\frac{p^m-1}{p-1}$ and $q_2=\frac{p-1}{2}$ are both prime. Let $g$ generate the quadratic residues group QR in $\F_{p^m}^*$ with order $\frac{p^m-1}{2}=q_1q_2$. On input $(p, m, g, g^{a}, g^b, g^c)$, one has to  distinguish between the two cases, where either $c = ab \mod q_1q_2$ or $c$ is uniformly random in $\Z_{q_1q_2}$.
\end{definition}

Classically, one may solve DDH in $\F_{p^m}^*$ by computing the discrete logarithms $a$ in QR, and check whether $(g^b)^a = g^c$. This takes superpolynomial time $L_{\frac 1 3}(p^m)$ using the Number Field Sieve. Notice that the quasi-polynomial algorithms of Barbulescu et al.~\cite{DBLP:conf/eurocrypt/BarbulescuGJT14} only apply in small characteristic. Alternatively, one may solve all three discrete logarithms $a,b,c$ in the subgroup~$G$ generated by $g^{q_1}$ of order $q_2$, and check whether $c = ab \mod q_2$. This takes time $\bigO(\sqrt{q_2}) = \bigO(\sqrt{p})$ using Pollard's Rho method. Thus, the best known classical algorithms require time $\min\{ L_{\frac 1 3}(p^m), \bigO(\sqrt{p}) \}$.

Quantumly, the discrete logarithm computation of $a$ in QR with either Shor, Eker\aa -H\aa stad, or Mosca-Ekert requires $m \log p$ output qubits to represent elements in QR. However, let us have a closer look at a quantum version of the second classical algorithm that works in $G$ of order $q_2$. The function $N$ that raises $g$ to the $q_1$-th power has the remarkable property that we {\em compress} the output via some group {\em homomorphism}, thus $N$ is a multiplicative hash function. Since $q_1=\frac{p^m-1}{p-1}$, we observe that $N$ is nothing but the multiplicative norm map
\[
  N:\F_{p^m}^* \rightarrow \F_p^*, \ x \mapsto x^{\frac{p^m-1}{p-1}}.
\]
Thus, $G = \F_p^*$ which admits representations of elements with only $\log p$ bits. As a consequence, hashing the output down via our homomorphic hash function $h=N$ saves us a factor of $m$ in the number of output qubits. Since Mosca-Ekert uses only a single input qubit, overall we go down from $1+m \log p$ to $1+ \log p$ qubits, again by a factor $\Theta(m)$. 

Our new quantum algorithm {\em first} classically hashes a discrete logarithm instance $(g,g^a)$ via $h=N$, and {\em second} applies quantum period finding via some function $f$. We show in the following that this efficiently realizes a universal embedding of $h \circ f$, where we first compute $f$ and then hash. Thus, it perfectly fits our hashing framework (besides the fact that we do not have a hash function family). Notice that this is a good example for usefulness of our hash technique even in the case of a single hash function, and without the need of a universal hash function family. 

Let us define the function 
\[
f_{g,g^a}:\Z\times \Z\mapsto \textrm{QR} \text{, with } (x,y) \mapsto g^{x}\cdot (g^a)^{y}\;
\]
with period $(a,-1)$. Thus, $f_{g,g^a}$ allows to solve the discrete logarithms in QR. 
Now observe that 
\begin{align*}
(h\circ f_{g,g^a})(x,y) &=h(f_{g,g^a}(x,y))
=h\left(g^{x}\cdot \left(g^a\right)^{y}\right)
=\left(g^{x}\cdot \left(g^a\right)^{y}\right)^{\frac{p^m-1}{p-1}}\\
&=\left(g^{\frac{p^m-1}{p-1}}\right)^x \cdot \left(\left(g^a\right)^{\frac{p^m-1}{p-1}}\right)^{y}
=h(g)^{x}\cdot \left(h\left(g^a\right)\right)^{y}\\
&=f_{h(g),h(g^a)}(x,y)
\end{align*}

Therefore, hashing the discrete logarithm instance $(g,g^a)$ via $h$ realizes a (very) efficient implementation of $h\circ f_{g,g^a}$, as desired.

\bibliographystyle{alpha}
\bibliography{literature}

\end{document}